\documentclass[%
 reprint,
 amsmath,amssymb,
 aps,
]{revtex4-2}
\usepackage[utf8]{inputenc}

\usepackage[caption=false]{subfig}
\usepackage{graphicx}
\usepackage{here}
\usepackage{epstopdf}
\usepackage{array}
\usepackage{physics}
\usepackage{verbatim}
\usepackage{amsmath,amsfonts,amssymb,amscd,mathtools}
\usepackage{amsthm}
\usepackage{tabularx}
\usepackage{stmaryrd}
\usepackage{enumerate}
\usepackage{wasysym}
\usepackage{braket}
\usepackage{mathrsfs}
\usepackage{microtype}
\usepackage{hyperref}
\usepackage{youngtab}
\usepackage[dvipsnames]{xcolor}
\hypersetup{
    bookmarksnumbered=true, % If Acrobat bookmarks are requested, include section numbers
    unicode=false, % non-Latin characters in Acrobat bookmarks
    pdfstartview={FitH}, % fits the width of the page to the window
    pdftitle={}, % title
    pdfauthor={}, % author
    pdfsubject={}, % subject of the document
    pdfcreator={}, % creator of the document
    pdfproducer={}, % producer of the document
    pdfkeywords={}, % list of keywords
    pdfnewwindow=true, % links in new window
    colorlinks=true, % false: boxed links; true: colored links
    linkcolor=NavyBlue, % color of internal links
    citecolor=NavyBlue, % color of links to bibliography
    filecolor=NavyBlue, % color of file links
    urlcolor=NavyBlue % color of external links
}
\usepackage{tikz}
\usepackage[lmargin=.7in,rmargin=.7in,tmargin=.7in,bmargin=1in]{geometry}
\usepackage{newtxtext,newtxmath}

\usepackage{cleveref}

\usepackage{algorithm}
\usepackage{algorithmicx}
\usepackage{algpseudocode}


\usepackage{booktabs}

\theoremstyle{plain}
\newtheorem{thm}{Theorem}

\newtheorem{lem}[thm]{Lemma}
\newtheorem{pro}[thm]{Proposition}

\theoremstyle{definition}
\newtheorem{defn}[thm]{Definition}

\newcommand{\eq}[1]{(\hyperref[eq:#1]{\ref*{eq:#1}})}

\renewcommand{\sec}[1]{\hyperref[sec:#1]{Section~\ref*{sec:#1}}}
\newcommand{\thrm}[1]{\hyperref[thrm:#1]{Theorem~\ref*{thrm:#1}}}
\newcommand{\lemm}[1]{\hyperref[lemm:#1]{Lemma~\ref*{lemm:#1}}}
\newcommand{\prop}[1]{\hyperref[prop:#1]{Proposition~\ref*{prop:#1}}}
\newcommand{\corr}[1]{\hyperref[corr:#1]{Corollary~\ref*{corr:#1}}}
\newcommand{\fig}[1]{\hyperref[fig:#1]{~\ref*{fig:#1}}}
\newcommand{\deff}[1]{\hyperref[deff:#1]{~\ref*{deff:#1}}}

\newcommand{\mE}{\mathcal{E}}

\newcommand{\mU}{\mathcal{U}}
\newcommand{\mT}{\mathcal{T}}
\newcommand{\mD}{\mathcal{D}}

\newcommand{\mH}{\mathcal{H}}

\newcommand{\mO}{\mathcal{O}}
\newcommand{\mB}{\mathcal{B}}

\newcommand{\mP}{\mathcal{P}}

\newcommand{\mbC}{\mathbb{C}}

\newcommand{\mbF}{\mathbb{F}}

\newcommand{\mbN}{\mathbb{N}}
\newcommand{\mbO}{\mathbb{O}}
\newcommand{\mbP}{\mathbb{P}}

\newcommand{\mbR}{\mathbb{R}}

\newcommand{\mbW}{\mathbb{W}}
\newcommand{\mbZ}{\mathbb{Z}}

\newcommand{\mfS}{\mathfrak{S}}

\newcommand{\mW}{\mathcal{W}}

\newcommand{\ve}{\varepsilon}

\DeclareMathOperator{\Span}{Span}
\DeclareMathOperator{\supp}{supp}

\DeclareMathOperator{\spec}{spec}

\newcommand{\ba}{\begin{eqnarray}}
\newcommand{\ea}{\end{eqnarray}}
\newcommand{\bann}{\begin{eqnarray*}}
\newcommand{\eann}{\end{eqnarray*}}
\newcommand{\bal}{\begin{equation}\begin{aligned}}
\newcommand{\eal}{\end{aligned}\end{equation}}
\newcommand{\dm}[1]{\ketbra{#1}{#1}}

\newcolumntype{L}[1]{>{\raggedright}p{#1}}
\newcolumntype{C}[1]{>{\centering}p{#1}}
\newcolumntype{R}[1]{>{\raggedleft}p{#1}}
\newcolumntype{D}{>{\centering\arraybackslash}X}

\usepackage{dcolumn}% Align table columns on the decimal point
\usepackage{bm}% bold math
\begin{document}

\title{Universal work extraction in quantum thermodynamics}% Force line breaks with \\
% \thanks{A footnote to the article title}%

\author{Kaito Watanabe}
\email{watanabe715@g.ecc.u-tokyo.ac.jp}
\affiliation{Department of Basic Science, The University of Tokyo, 3-8-1 Komaba, Meguro-ku, Tokyo 153-8902, Japan}

\author{Ryuji Takagi}
\email{ryujitakagi@g.ecc.u-tokyo.ac.jp}
\affiliation{Department of Basic Science, The University of Tokyo, 3-8-1 Komaba, Meguro-ku, Tokyo 153-8902, Japan}

%\date{\today}% It is always \today, today,
             %  but any date may be explicitly specified

%TC:ignore
\begin{abstract}
Evaluating the maximum amount of work extractable from a nanoscale quantum system is one of the central problems in quantum thermodynamics. Previous works identified the free energy of the input state as the optimal rate of extractable work under the crucial assumption: experimenters know the description of the given quantum state, which restricts the applicability to significantly limited settings. Here, we show that this optimal extractable work can be achieved without knowing the input states at all, removing the aforementioned fundamental operational restrictions. We achieve this by presenting the construction of a quantum channel whose description does not depend on input states but nevertheless extracts work quantified by the free energy of the unknown input state. Remarkably, our result partially encompasses the case of infinite-dimensional systems, for which optimal extractable work has not been known even for the standard state-aware setting. Our results clarify that, even though whether the description of the given state is provided at the beginning of the protocol generally makes the operational setting fundamentally different in accomplishing information-theoretic tasks, not knowing the input state does not influence the optimal performance of the asymptotic work extraction.
\end{abstract}
%TC:endignore
\maketitle

\vskip\baselineskip
\begin{center}
{ \bf INTRODUCTION}
\end{center}

Evaluating the amount of work that can be extracted from a given quantum state is one of the pivotal problems in the framework of quantum thermodynamics. 
Specifically, the recent progress in the realm of nanoscale technology allows us to manipulate microscopic quantum systems, and it is of fundamental importance to understand the thermodynamical properties of the microscopic system governed by the laws of quantum mechanics.
A series of recent works have made much progress in characterizing the possible amount of work extracted from quantum states employing quantum information theoretic approaches~\cite{horodecki_fundamental_2013, Faist_fundamental_work, Brandao_2015_second_law, Brandao_resource_theory_of, lostaglio_quantum_coherence, Gour_role_of}.
These studies have found that, in the scenario where many copies of a quantum state in contact with a heat bath are given, the optimal extractable work is characterized by the Helmholtz free energy in the asymptotic limit~\cite{Brandao_resource_theory_of, Gour_role_of}.

Although previous works have established important insights into the fundamental limit on extractable work, the scenarios considered do not reflect the natural operational setting.
In particular, they assume that the experimenter is given a complete description of the input state in advance, which allows one to tailor the work extraction protocol depending on the details of the state.
However, since the input state might be prepared through a highly complex quantum process such as a deep quantum circuit or exposed to some unknown noise, a classical description of the quantum state in hand is usually not accessible.

Although one could run the quantum state tomography~\cite{Vogel_1989_Determination, Banaszek_2013_focusing} to obtain the information of the initial state, state tomography needs to consume a large number of copies of the state. 
This is not merely a technological issue from a practical viewpoint but is a fundamental one. 
The performance of work extraction is evaluated by how much work one could extract from the given state per copy. Therefore, consuming many copies for state tomography can significantly reduce the rate of extracted work.
Furthermore, it is unclear whether measurements required for state tomography can be realized ``for free''~\cite{horodecki_fundamental_2013}, as the measurement process itself could cost work~\cite{Sagawa2009minimal} and could further reduce the performance of work extraction.

These difficulties, associated with the fundamental thermodynamic cost in learning information of the given state, suggest the necessity of a novel framework that can deal with \emph{state-agnostic} scenarios---where the description of a quantum channel representing the protocol cannot depend on the input state.
Recent works have made the initial step in this direction by considering the manipulation of partially characterized states.
Ref.~\cite {boes_statistical_2018} clarified the possible state transformation with processes that conserve average energy, provided partial information about the energy of the given state.
In another approach, Ref.~\cite{xuereb2024resources} addressed the state-agnostic scenario by employing the probe with a limited dimension, and analyzed the performance of thermodynamic tasks with the partial information.
Especially for the work extraction task, Ref.~\cite{Dominik_work_extraction} considered the ergotropy accessible after coarse-grained measurements and characterized this with respect to observational entropy, followed by the analysis of sufficient sample cost~\cite{chakraborty_2024_sample_complexity}.
Ref.~\cite{Watanabe_Black_box} formalized the partial information in the form of ``black box'', the known subset of quantum states from which the given state is promised to belong to, and characterized the guaranteed extractable work.  
These results would indicate that extractable work from the given quantum state could crucially depend on the information of the state provided at the beginning of the protocol.

Here, we defeat this conjecture in the most significant and positive manner.
We establish the \emph{universal} work extraction protocol, whose construction is independent of the input state, that nevertheless achieves the asymptotic work extraction rate characterized by the free energy of the input state. 
Recalling that free energy is the optimal rate for the standard ``state-aware'' setting, our results show that the optimal work extraction can be realized regardless of the preknowledge of the given state.

We further deal with the case where the system is infinite-dimensional, the primal examples of which include the quantum optical setting involving the bosonic Hamiltonian.
To the best of our knowledge, the optimal work extraction rate for the infinite-dimensional system has not been established even for the state-aware setting---the free energy was shown to be an upper bound of the optimal rate~\cite{ferrari_asymptotic_2023}, but whether this bound could be achieved was not known even under the state-aware scenario.
We show that the optimal work extractable from the known input state can be characterized by the Helmholtz free energy under a natural physical assumption.
Moreover, we construct a semiuniversal protocol that can extract the maximum work of the free energy for any state picked from the set consisting of a finite number of states.

Our results provide far-reaching insights beyond quantum thermodynamics. The work extraction task can be seen as a special case of the \emph{resource distillation} task in quantum resource theories~\cite{Chitamber_2019_QRT, gour_2024_resources_quantum}, a general framework to analyze the resourcefulness of the quantum states and quantum processes. In this sense, our result offers a richer landscape of the resource distillation tasks in state-agnostic scenarios. 
Indeed, for entanglement distillation, it was shown that universal entanglement distillation from unknown pure states can be achieved with the optimal rate~\cite{matsumoto_universal_distortion}. 
We show that, despite the difference between the structures of quantum thermodynamics and entanglement, universal work extraction is indeed possible, suggesting the potential of further extending the notion and technique of universal resource distillation to various kinds of quantum resources.

\begin{figure}
    \centering
    \includegraphics[width=1\linewidth]{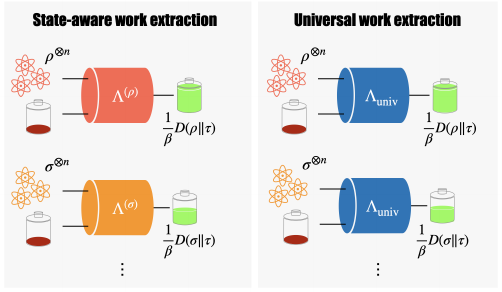}
    \caption{Schematic figure of the universal work extraction protocol. The work extraction protocol discussed in the previous results (left) is tailored according to the information of the initial state, and the optimal extractable work rate is shown to be characterized by the Helmholtz free energy. The universal work extraction protocol introduced in our result (right) is independent of the input state but achieves the same extractable rate as the state-dependent protocol in the asymptotic limit.}
    \label{fig: universal work extractor}
\end{figure}

\vskip\baselineskip
\begin{center}
{ \bf  RESULTS}
\end{center}

\begin{center}
{ \bf Preliminaries}
\end{center}

We consider a quantum system in contact with a thermal bath of inverse temperature $\beta$ associated with a Hilbert space $\mH$, which does not have to be finite-dimensional, and a Hamiltonian $H$ on $\mH$.

When the $\mH$ is infinite-dimensional, we restrict our attention to the Hilbert space spanned by the energy eigenstates $\qty{\ket{i}}_{i\in\mbN}$ of the  Hamiltonian $H=\sum_{i}E_i\ketbra{i}{i}$. Without loss of generality, we can rearrange the order of the energy eigenvalues so that $E_i\leq E_{i+1}$ holds for any $i\in\mbN$.
We assume that the Hamiltonian satisfies the Gibbs hypothesis $\Tr[e^{-\beta H}]<\infty$ so that the partition function is well-defined.

We employ the resource-theoretic approach, which has recently made significant progress in investigating the thermodynamic features of quantum systems~\cite{horodecki_fundamental_2013, Faist_fundamental_work, Brandao_2015_second_law, Brandao_resource_theory_of, Gour_role_of, Lostaglio2019introductory, lostaglio_quantum_coherence, Faist_2019_thgermodynamic_capacity, shiraishi_sagawa_QT_correrated, tajima_2024_GPO, janzing_thermodynamic_2000, sagawa_asymptotic_2021, shiraishi_2024_quantumthermodynamics}, to analyze the work extraction.
Quantum resource theory is characterized by the set of states prepared without any cost in the setting (called \emph{free states}) and the class of operations that can be applied easily (called \emph{free operations}).
In the context of quantum thermodynamics, the Gibbs thermal state defined as $\tau:=e^{-\beta H}/\Tr[e^{-\beta H}]$ is the only free state.
We consider the class of operations called the thermal operations~\cite{horodecki_fundamental_2013}, which is considered a physically implementable class of operations. 
Let $\mD(\mH)$ denote the set of density operators acting on the Hilbert space $\mH$.
A completely positive trace preserving (CPTP) map $\mE:\mD(\mH_A)\to\mD(\mH_B)$ is called a thermal operation if there exists an ancillary system $E$ with the Hamiltonian $H_E$ such that $\mE$ can be dilated as
\bal
\mE(\rho_A)=\Tr_{A+E-B}\qty[U(\rho_A \otimes \tau_E)U^\dagger],
\eal
where $U$ is the unitary operator that conserves the energy of the composite system $A+E$, that is, satisfies $[U, H_A\otimes I_E+I_A\otimes H_E]=0$. 
We remark that any thermal operation $\mE$ is Gibbs-preserving, i.e., it holds that $\mE(\tau_A)=\tau_B$.

When we consider the multi-copies of the systems, we assume that the Hamiltonian of each system is the same, and there are no interactions between the systems, that is, the Hamiltonian of the $n$ systems are represented as $H^{\times n}:=\sum_{i=1}^n I^{\otimes (i-1)}\otimes H\otimes I^{\otimes (n-i)}$.
The thermal state of the whole system is described as $\tau^{\otimes n}$.

The work that can be extracted from the input state is measured with the energy gap of another system called the \emph{work storage}~\cite{horodecki_fundamental_2013}, a qubit system $\mH_W=\Span\qty{\ket{0},\ket{W}}$ with Hamiltonian $H_W=W\ketbra{W}{W}$.
Starting from a given input state $\rho$ and the ground state $\ketbra{0}{0}$ of the work storage, if there exists a free operation $\Lambda$ which can output the state $\ketbra{W}{W}$ with a target fidelity error $\ve$, i.e., $\Lambda(\rho\otimes \ketbra{0}{0})\approx^\ve\ketbra{W}{W}$ holds, we can conclude that the work $W$ is extracted from $\rho$ with error $\ve$.
The optimal one-shot extractable work from the state $\rho$ in the state-aware scenario is defined as 
\bal
&W^\ve_{\rm aware}(\rho)\\ &:=\max\qty{W\in\mbR~|~\sup_{\Lambda\in {\rm TO}}F(\Lambda(\rho\otimes \ketbra{0}{0}),\ketbra{W}{W}) \geq 1-\ve},
\eal
where the supremum is taken over all possible work extraction protocols that are thermal operations. This definition is legitimate in the sense that thermal operation does not require any additional work cost to apply.
What is crucial in the context of the following discussion is that the optimal protocol $\Lambda$ that achieves the maximum work extraction can generally depend on the input state $\rho$.

To see the correspondence between the framework of quantum thermodynamics and thermodynamics in the macroscopic regime, it is important to evaluate the performance of the work extraction task in the multi-copy limit. (See FIG.~\ref{fig: universal work extractor}.)
The asymptotic extractable work rate is defined as
\bal
W_{\rm aware}^\infty(\rho)=\lim_{\ve\to 0}\limsup_{n\to\infty}\frac{1}{n}W_{\rm aware}^\ve(\rho^{\otimes n})
\eal
In Ref.~\cite{ferrari_asymptotic_2023}, they showed that the upper bound on the extractable work extraction rate is given by 
\bal\label{eq: converse}
\beta W^\infty_{\rm aware}(\rho)\leq D(\rho\|\tau)
\eal
even in the situation where the Hilbert space under consideration is infinite-dimensional.
Here, $D(\rho\|\tau)$ is the Umegaki relative entropy defined as $D(\rho\|\tau):=\Tr[\rho\log\rho-\rho\log\tau]$~\cite{umegaki_conditional_1954}.
Umegaki's relative entropy plays a pivotal role in quantum thermodynamics because it connects to Helmholtz's free energy.
However, whether this limit is achievable in the infinite-dimensional case was unknown.

A series of papers have found that this bound is tight in the finite-dimensional system, i.e., the optimal rate of extractable work from an i.i.d. state in the finite-dimensional system per the number of copies by the thermal operation is characterized as~\cite{Brandao_2015_second_law, Gour_role_of}
\bal\label{Eq: asymptotic aware free}
\beta W_{\rm aware}^\infty(\rho)=D(\rho\|\tau).
\eal

\begin{center}
{ \bf Universal work extraction}
\end{center}

Since we do not know the input state in the state-agnostic scenario, we also do not know how much work is supposed to be extracted, and how large the energy gap the work storage should have.
To avoid this problem, we extend the work storage as follows. First, we define a set $\mbW$ as
\bal
\mbW=\qty{\sum_{i=1}^d N_i E_i \geq 0~|~\sum_i N_i=0, ~~N_i\in\mbZ}.
\eal
This set includes all the possible work extracted from the system, which can be seen from the construction of the protocol.
The Hamiltonian of work storage is defined as
\bal
H_W=\sum_{W\in\mbW}W\ketbra{W}{W}_W.
\eal
Thanks to this definition, the work storage can admit any possible amount of work.

The difference between state-aware and agnostic work extraction is whether the distillation process can depend on the given state. To formalize this, we define one-shot extractable work of the input state $\rho$ with error $\ve>0$ enabled by the thermal operation $\Lambda$  as
\bal
{}&W^\ve(\rho,\Lambda)\\
&=\max\qty{W\in\mbW~|~F(\Lambda(\rho^{\otimes n}\otimes \ketbra{0}{0}_W),\ketbra{W}{W}_W)\geq 1-\ve}.
\label{eq:one-shot work protocol dependent}
\eal
We then define the asymptotic extractable work rate of the sequence of states $\qty{\rho^{\otimes n}}_{n\in\mbN}$ enabled by the protocol $\qty{\Lambda_n}_{n\in\mbN}$ as
\bal
{}&W^{\infty}(\qty{\rho^{\otimes n}}_n,\qty{\Lambda_n}_{n})=\lim_{\ve\to+0}\limsup_{n\to\infty}\frac{1}{n}W^\ve(\rho^{\otimes n},\Lambda_n).
\label{eq:asymptotic work protocol dependent}
\eal

If the series $\{\Lambda_n\}_n$ of the protocol can be optimally chosen depending on the state $\rho$, this recovers the state-aware work extraction, i.e.,

\bal
W_{{\rm aware}}^\infty(\rho):=\sup_{\{\Lambda_n\}_n\subset {\rm TO}}W^{\infty}\qty(\qty{\rho^{\otimes n}}_n,\qty{\Lambda_n}_{n}),
\eal
which is characterized by the free energy of $\rho$ as in \eqref{Eq: asymptotic aware free}.

 On the other hand, state-agnostic work extraction requires us to fix the protocol first and see how well it works for different input states. Therefore, the notion of state-agnostic extractable work should be considered as the function of all quantum states such that there is a fixed protocol that works for all states with that performance.
\begin{defn}\label{def:agnostic}
Fix a set $S\subset \mD(\mH)$ of states. 
   If there exists a series $\{\Lambda_n\}_n$ of thermal operations $\Lambda_n:\mD(\mH^{\otimes n}\otimes \mH_W)\to\mD(\mH_W)$ such that 
\bal
W^{\infty}\qty(\qty{\rho^{\otimes n}}_n,\qty{\Lambda_n}_{n}) = W_{\rm agnostic}^\infty(\rho),\quad \forall \rho\in S,
\eal
we say that the function $W_{\rm agnostic}^\infty:S\to \mbR$ is a $S$-achievable state-agnostic work extraction rate. 
When we take $S=\mD(\mH)$, we say that $W^\infty_{\rm agnostic}$ is the achievable state-agnostic work extraction rate.
\end{defn}
$S$ represents the possible candidates for the given state, and the condition $S=\mD(\mH)$ implies that we cannot utilize any information of the given state to tailor the work extraction protocol.
We remark that this notion of state-agnostic (or interchangeably, \emph{universal}) protocols was previously considered in the context of quantum source compression~\cite{JozsaPresnell2003,Hayashi2002quantum,Bennett2006universal} and entanglement distillation~\cite{matsumoto_universal_distortion,leone_entanglement_theory_limited_computational}, in which the partial information of the given state may be provided.

Because of the result for state-aware work extraction \eqref{Eq: asymptotic aware free} for finite dimensions, any achievable state-agnostic work extraction rate $W_{\rm agnostic}$ for finite dimensions satisfies 
\bal
\beta W_{\rm agnostic}^\infty(\rho)\leq D(\rho\|\tau),\quad \forall \rho\in\mD(\mH).
\eal
Therefore, the best we can hope for is to have an achievable $W_{\rm agnostic}^\infty$ such that the equality holds for an arbitrary state $\rho$. 
Our main result is that this is indeed the case.

\begin{thm}\label{thm:universal}
 The state-agnostic work extraction rate $W_{\rm agnostic}^\infty$ such that $\beta W_{\rm agnostic}^\infty(\rho)=D(\rho\|\tau)$ for all state $\rho$ in the finite-dimensional system is achievable.
\end{thm}

The main idea behind our protocol is to utilize the permutational symmetry of the given copies of the unknown states, which allows us to circumvent learning the full description of the given quantum state.
More details are given in the next section and the Appendix~\ref{app sec: universal work extraction}.

 Let us also remark on the relationship between the universal work extraction and Maxwell's demon. 
In the standard setting of Maxwell's demon thought experiments, we assume that the experimenter knows the probability distribution of the system, but does not know which state is actually realized.
Maxwell's demon is a hypothetical agent that acquires knowledge of which state is realized, enabling it to apply a feedback protocol and extract some amount of work from an equilibrium system, seemingly violating the second law of thermodynamics. 
This thought experiment shows that knowledge of the state increases the performance of work extraction, which apparently contradicts our result. 
However, these two results do not contradict each other.
Our setting is one in which the experimenter does not even know the density matrix, whereas the state-aware scenario involves prior knowledge of it. Thus, the object to which the word ``prior knowledge'' refers differs between Maxwell's demon setting and the state-agnostic work extraction in our result.

Let us discuss a relation with the recent result in Ref.~\cite{Watanabe_Black_box}, which discussed the state-agnostic work extraction by introducing the black box work extraction, where the worst-case extractable work among all states in the given set of states is considered. 
They found that, when the black box only contains a finite number of states, the performance of the black box work extraction under thermal operations is characterized by the minimum free energy of the states in the box. 
Our result extends this to an arbitrary black box composed of i.i.d. states, solving the open problem raised in Ref.~\cite{Watanabe_Black_box}.

Furthermore, our result enables us to analyze the performance of state-agnostic work extraction from any given state, which is not possible in the framework of black box work extraction because it always considers the worst-case performance.
Specifically, the extractable work from any black boxes containing the thermal state is always zero, because no work can be extracted from the thermal state.
Thus, this does not reflect the properties of all the states in the black box that might be given to us.
On the other hand, our new result fully characterizes the performance of the work extraction protocol for any given state.

Note that we cannot know the amount of extractable work from the state after we apply this protocol. When we utilize the work storage for some practical tasks, we need to perform the projective measurement with the set of the projectors $\qty{\ketbra{W}{W}_W}_{W\in\mbW}$. This process is considered beyond the framework of thermal operations, as it could generally change the energy of the work storage. 
Nevertheless, it only changes the energy of the final state by a small amount because the final state is sufficiently close to the energy eigenstate.

We also remark that the universal resource distillation is tied to the notion called \emph{pseudo-resource states}~\cite{aaronson_2023_quantumpseudoentanglement,Gu_pseudomagic,arnonfriedman2023computationalentanglementtheory,leone_entanglement_theory_limited_computational}---state ensembles that cannot be efficiently distinguished from more resourceful ones---because a state-agnostic resource distillator could be used as a state distinguisher. 
Therefore, our protocol could provide useful insights in investigating the pseudo-resource in the framework of quantum thermodynamics. (See Appendix~\ref{app sec: finite copy and pseudo} for more discussions.)

Let us now discuss the performance of the work extraction task for infinite-dimensional systems.
The following results show that, when the system is associated with an infinite-dimensional Hilbert space, one can still construct a work extraction protocol that does not depend on the full details of the given input state, while achieving the optimal work extraction rate.

\begin{thm}\label{thm: inf universal}
    Let $S\subset \mD(\mH)$ be a set of states that contains a finite number of states.
    Furthermore, we assume that all the states have finite energy and free energy, and, for any $\rho\in S$, there exists a positive number $\ve>0$ such that the diagonal elements of $\rho$ satisfy $\rho_{ii}=\mO(i^{-(2+\ve)})$.
     Then, the state-agnostic work extraction rate $W_{\rm agnostic}^\infty$ such that $\beta W_{\rm agnostic}^\infty(\rho)=D(\rho\|\tau)$ in the infinite-dimensional system is $S$-achievable.
\end{thm}
Several remarks are in order.
We first would like to clarify that our result can only be applied to the case when the number of candidate states is finite.
In this sense, this protocol can be seen as a \emph{semi-universal} work extraction protocol, which functions universally for a limited set of states. 
Whether one can extend this to a continuous set of input states (e.g., the set of all density matrices) is an interesting open problem.
We also remark that we used an additional assumption on the scaling behavior of the diagonal elements of the input state. 
Although this limits the scope of the universality of the protocol compared to the finite-dimensional cases, our protocol still does not depend on most of the information about the input state, such as actual values of the diagonals or its eigenbasis.

We stress that Theorem~\ref{thm: inf universal} offers a novel result even if we take the simplest case $S=\qty{\rho}$, i.e., $S$ is a singleton, which corresponds to the state-aware scenario.
This, together with the converse bound in Eq.~\eqref{eq: converse}, gives us the complete characterization of the optimal extractable work rate for the states in the infinite-dimensional system satisfying the condition, i.e., the optimal extractable work can be characterized by
\bal
\beta W^{\infty}_{\rm aware}(\rho)=D(\rho\|\tau).
\eal
Theorem~\ref{thm: inf universal} then further shows that this optimal rate can be achieved in a state-agnostic manner if we know that the state is taken from a finite number of candidates.

Let us remark on the assumption made here. 
If we take the Hamiltonian as the harmonic oscillator, the condition $\rho_{ii}=\mO(i^{-(2+\ve)})$ is almost equivalent to the condition of finite energy.
If the energy spectrum grows superlinearly $E_i=\Omega(i^\alpha),~(\alpha>1)$, the condition $\rho_{ii}=\mO(i^{-(2+\ve)})$ is ensured by the finite-energy condition.

Our results particularly apply to bosonic systems, one of the most essential ones toward realizing quantum computing that involves an infinite-dimensional Hilbert space.
When we analyze the quantum thermodynamic properties of the states in the bosonic system, what matters is not only the nonequilibriumness, but also the non-Gaussianity, which serves as another important resource in the bosonic system~\cite{Genoni_Quantifying_nG_character,Genoni_quantifying_nG_for_QI, marian_2013_relative_entropy,Zhuang_RT_of_nG,Takagi_convex_resource,Albarelli2018resource}. 
Motivated by this observation, previous works introduced the framework of \emph{Gaussian thermal operations}~\cite{Serafini_gaussian_thermal, yadin_catalytic_2022, narasimhachar_thermodynamic_2021}, which is the intersection of Gaussian and thermal operations. 
(See also Ref.~\cite{rodriguez_2025_extractingenergybosonicgaussian} for the analysis of the extractable work from the bosonic system using Gaussian unitaries in the sense of the ergotoropy.)
It is then natural to ask whether our optimal rate could be realized by Gaussian thermal operations.

Interestingly, the semiuniversal work extraction protocol needs to be non-Gaussian if the set $S$ includes a Gaussian state, no matter what Fock state we take for the initial work storage state. 
Such a semiuniversal work extraction needs to convert the input Gaussian non-thermal state and the initial Fock state in the work storage to the target Fock state.
However, we can see that this cannot be achieved by Gaussian operation by looking at a measure of non-Gaussianity. 
In particular, the negativity of the Wigner function~\cite{Kenfack2004negativity} was shown to be a valid non-Gaussianity measure, which cannot increase under Gaussian operations (in fact, under a more general class called Gaussian protocols)~\cite{Takagi_convex_resource,Albarelli2018resource}. 
Since the negativity of the Wigner function of Fock states monotonically increases with energy, the output state of the work extraction protocol has higher non-Gaussianity than the initial state, excluding the possibility of Gaussian thermal operations.

\begin{center}
{ \bf Sketch of construction}
\end{center}

Let us briefly overview the construction of our universal work extraction protocol.
The state-aware work extraction protocol in Ref.~\cite{Brandao_resource_theory_of}, which is also explained in detail in Appendix~\ref{App Sec: Brandao}, goes as follows: Apply the pinching channel, which corresponds to taking the time-average, apply the energy-conserving unitary to diagonalize the input state with a fixed energy eigenbasis, which enables us to apply the energy-conserving unitary to extract work.
Here, the state-dependent steps in this protocol are diagonalization and the appropriate choice for the unitary to extract work.
Our proof strategy is to convert these steps to state-independent procedures.

The overview of the universal work extraction protocol functions in the following steps. These steps are also exhibited in FIG~\ref{fig: overview of protocol}.
\begin{figure}
    \centering
    \includegraphics[width=1\linewidth]{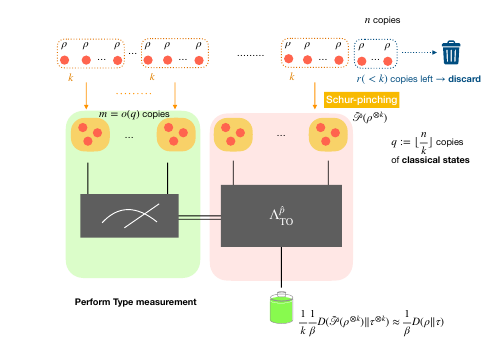}
    \caption{Overview of the universal work extraction protocol for finite-dimensional systems. First, we apply the channel called Schur pinching to obtain the state diagonalized with a specific energy eigenbasis that also respects the permutation symmetry. After that, we apply a thermal operation that simulates the type measurement on a sublinear number of subsystems and the work extraction protocol conditioned on the measurement outcomes.
    Since the projector corresponding to the measurement is the projector onto the energy eigenspace, we can realize the same action solely by a thermal operation.}
    \label{fig: overview of protocol}
\end{figure}
\begin{enumerate}
    \item (Diagonalization step) Given $n$ copies of $\rho$, we apply the channel called Schur pinching channel, which is explained in the following discussion, to each $k$ copies.
    \item (Learning step) Estimate the relative entropy of the pinched state with the incoherent measurement. As explained below, this procedure, combined with the execution step, can be done solely by the thermal operations. 
    \item (Execution step) According to the information about the relative entropy, we apply the state-aware work extraction protocol constructed in Ref.~\cite{Brandao_resource_theory_of}.   
\end{enumerate}
Since all the steps can be done by thermal operations, the concatenation of these procedures is also a thermal operation.

Let us first consider the diagonalization part.
The most naive idea to obtain a state diagonalized with a fixed energy eigenbasis is to choose an arbitrary energy eigenbasis $\qty{\ket{i}}_{i=1}^{d^n}$, and apply the completely decohering map $\Delta(\rho_n):=\sum_{i=1}^{d^n}\ketbra{i}{i}\rho_n\ketbra{i}{i}$. 
However, the decohering map might lose too much free energy. 
For instance, if we choose the energy eigenbasis as the tensor products $\qty{\ket{E_{i_1}}\otimes \cdots\otimes \ket{E_{i_n}}}$ of the energy eigenvectors of each subsystem, we have $\Delta(\rho^{\otimes n})=(\Delta(\rho))^{\otimes n}$. Since the performance of the state-agnostic work extraction is always upper bounded by that of state-aware protocol, the extractable work from this system is at most $\frac{1}{\beta}D(\Delta(\rho)\|\tau)$, which is smaller than the optimal extractable work $\frac{1}{\beta}D(\rho\|\tau)$ in general. 
This is due to the nondiagonal entries of the density matrix that involve energetic coherence.
Thus, we need to find some way to obtain the diagonalized state without losing too much coherence.

To this end, we utilize the permutation symmetry of the system and the input state.
First, we consider the Hilbert space $\mH^{\otimes k}$ of $k$ systems. Due to the Schur-Weyl duality, this Hilbert space can be decomposed as follows.
\bal\label{Eq: Schur-Weyl duality of the space}
\mH^{\otimes k}=\bigoplus_{\lambda\in Y_d^k}\mW_\lambda\otimes \mU_\lambda.
\eal
Here, $Y^k_d$ is the set of the Young diagrams of $k$ blocks with depth at most $d$, and $\mW_\lambda$ and $\mU_\lambda$ are the representation spaces of Weyl representation of the general linear group ${\rm GL}(d,\mbC)$ and the irreducible representation (irrep) of the symmetric group $\mfS_k$ respectively, corresponding to the Young diagram $\lambda$.
Here, we denote the dimension of the representation spaces as $m_\lambda=\dim\mU_\lambda,~ n_\lambda=\dim\mW_\lambda$.

Since $H^{\times k}$ is permutationally invariant, it can be decomposed as
\bal
H^{\times k}=\bigoplus_{\lambda\in Y_d^k}H_\lambda\otimes I_{\mU_\lambda}.
\eal
\begin{figure}
    \centering
    \includegraphics[width=\linewidth]{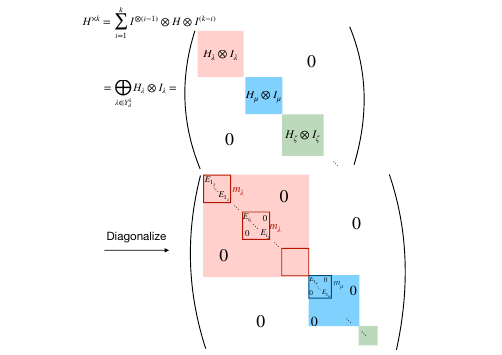}
    \caption{The structure of the Hamiltonian with the Schur basis. Due to the Schur-Weyl duality, the Hamiltonian $H^{\times k}$ of $k$ systems is block-diagonalized as above. Since all $H_\lambda$ are Hermitian, we can find an orthonormal basis of the whole Hilbert space that diagonalizes $H_\lambda$ in each block.}
    \label{fig: Hamiltonian Schur basis}
\end{figure}
Since $H_\lambda$ is Hermitian for any $\lambda\in Y^k_d$, each $H_\lambda$ can be diagonalized by some orthogonal basis. This forms the energy eigenbasis of $H^{\times k}$ (FIG.~\ref{fig: Hamiltonian Schur basis}).
$\rho^{\otimes k}$ can also be represented as
\bal
\rho^{\otimes k}=\bigoplus_{\lambda\in Y_d^k}\rho_\lambda\otimes I_{\mU_\lambda}.
\eal
The explicit form of $\rho^{\otimes k}$ is exhibited as FIG.~\ref{fig: Schur basis}.

\begin{figure}
    \centering
    \includegraphics[width=1\linewidth]{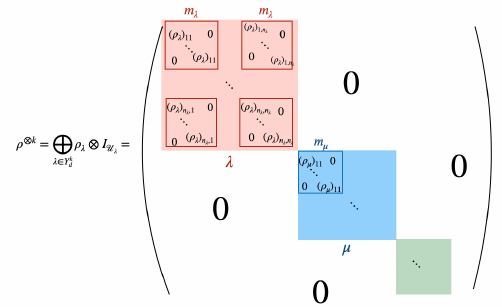}
    \caption{The matrix representation of $\rho^{\otimes k}$. The blocks indicate each direct sum element $\mW_\lambda\otimes \mU_\lambda$. Each block consists of $n_\lambda\times n_\lambda$ blocks $(\rho_\lambda)_{ij}I_{m_\lambda}$, where $m_\lambda=\dim \mU_\lambda$ and $n_\lambda=\dim\mW_\lambda$. }
    \label{fig: Schur basis}
\end{figure}

We then apply the pinching channel defined as
\bal
\tilde{\mP}(\cdot)=\sum_{\lambda\in Y^k_d} \sum_{E_{i_\lambda}}\Pi_{E_{i_\lambda}}(\cdot)\Pi_{E_{i_\lambda}}.
\eal
This channel---which we call \emph{Schur pinching}---removes the terms in the off-diagonal blocks in the structure induced by the Schur basis (FIG.~\ref{fig: Schur basis}).
Here, $\Pi_{E_{i_\lambda}}$ is the projector onto the eigenspace of $H_\lambda\otimes I_{\mU_\lambda}$ which corresponds to the energy eigenvalue $E_{i_\lambda}$. 
Lemma~S.1 in Appendix~\ref{app sec: lemmas} ensures that this channel is a thermal operation.
This channel corresponds to the procedure to remove all the non-diagonal blocks within each irrep $\mW_\lambda\otimes \mU_\lambda$. 
The output state after applying this channel is the diagonal state with the fixed Schur basis.
We provide an explicit construction of Schur pinching in the 3-qubit case in Section III.B.

Analysis in Appendix~\ref{app sec: universal work extraction} reveals that, if we take sufficiently large $k$, The relative entropy $\frac{1}{k}D(\tilde{\mP}(\rho^{\otimes k})\|\tau^{\otimes k})$ of the pinched state can be arbitrary close to the relative entropy $D(\rho\|\tau)$ of the input state.
Therefore, if we can design the universal work extraction protocol for the classical state that can achieve the relative entropy of the input state, we can take $k$ sufficiently large so that we achieve the extractable work rate $D(\rho\|\tau)$. 

We can now concentrate on designing the universal work extraction protocol for classical states.
As also mentioned in the previous discussion, it is nontrivial to choose the right unitary to extract work from classical states in the state-agnostic scenario.
From the design of the work extraction protocol in Ref.~\cite{Brandao_resource_theory_of}, it is sufficient to have the knowledge of the relative entropy of the state with respect to the thermal state to choose the appropriate protocol.

However, the class of thermal operations does not contain any measurement, preventing us from performing tomography. To avoid this problem, we consider a slightly larger class of operations called conditioned thermal operations~\cite{Narasimhachar_CTO}, which consists of the measurement on the subsystem and the thermal operation conditioned by the measurement outcome. 
Although the conditioned thermal operations outperform the thermal operations, it is shown in Ref.~\cite{Watanabe_Black_box} that when the measurement is restricted to the incoherent projective measurement, i.e., the projective measurement whose POVM elements $M$ satisfy $\mP(M)=M$, the class of conditioned thermal operations coincides with the class of thermal operations.
The idea of implementation is to convert the incoherent measurement and the following energy-conserving unitary conditioned on the estimation to a controlled energy-conserving unitary gate.
Since we convert the input state to the classical state $\tilde{\mP}(\rho^{\otimes k})$, which is diagonalized with the fixed energy eigenbasis, we can apply the incoherent measurement to extract sufficient information to estimate the relative entropy.

Now, the whole description of the protocol is the following. 
Given $n$ copies of the unknown state $\rho$, we apply the Schur pinching to $k$ copies. 
After this, we obtain $q:=\lfloor\frac{n}{k}\rfloor$ copies of the classical state $\tilde{\mP}(\rho^{\otimes k})$. 
We discard the remaining $r:=n-kq$ systems.
Out of $q$ copies of $\tilde{\mP}(\rho^{\otimes k})$, pick up $m=o(q)$ copies and apply the type measurement on them. 
We estimate the relative entropy according to the measurement outcome and apply the work extraction protocol in Ref.~\cite{Brandao_resource_theory_of}.

The parameters $k$ and $m$ can depend on the total number of input copies $n$.
In Appendix~\ref{app sec: universal work extraction}, we show that we can tune these parameters $k_n, m_n$ so that this work extraction protocol can achieve the extractable work rate $\frac{1}{\beta}D(\rho\|\tau)$ in the asymptotic limit.

One might not be convinced to call the protocol constructed in our paper ``universal'' since our protocol contains the learning process that extracts information about the relative entropy of the input state. 
We first stress that, although our protocol involves conditional processes, the quantum channel that describes the whole process, transforming the initial states $\rho^{\otimes n}$ to the work state $\dm{W}$, is chosen independently of the initial state $\rho$.
This means that our protocol can be universally applied to any i.i.d. input states, hence the name. (See Definition~\ref{def:agnostic} and Eqs.~\eqref{eq:one-shot work protocol dependent} and \eqref{eq:asymptotic work protocol dependent}.) 

Nevertheless, the incorporation of the learning process does not make our universal work extraction protocol trivial, due to the tradeoff relation between the copies of states for tomography and the accuracy of the estimation.
To achieve the optimal work extraction rate, we can only consume a sublinear amount of input states for the state tomography, and thus we can never obtain the information of the input state with arbitrary accuracy.
On the other hand, the execution protocol subsequently applied after learning might need more accuracy than that ensured by performing tomography on a sublinear number of subsystems. In that case, the fidelity error might grow up to 1 in the limit $n\to\infty$.
Whether this issue arises or not depends on the specifics of the protocol and its performance analysis, and what our result shows precisely is that one can design a work extraction protocol that can avoid this problem.

We also exhibit two alternative constructions of the universal work extraction protocol: one based on the measure-and-prepare strategy and the other on a more naive tomography-based strategy, in Appendix~\ref{app sec: measure and prepare}.

Let us briefly remark on the performance of the universal work extraction protocol from the finite number of input states, i.e., how fast the work extraction rate approaches to the optimal rate with respect to the number of input-state copies. (See Appendix~\ref{app sec: finite copy} for details.)
We find that our universal protocol entails a difference from the state-aware protocol in this regime, where the former shows a much slower convergence than the latter. 
On the other hand, consider a slightly relaxed scenario where we are informed of the relative entropy of the input state, while not knowing the other description of the density matrix. 
This assumption, where one is given the entropic quantity of given quantum states or channels, is often made in the context of the universal protocols in quantum Shannon theory~\cite{Josza1998universal, Hayashi_2009}.
In this setting, we find that our universal protocol with Schur pinching achieves the convergence speed that coincides with that of the state-aware protocol.
% Indeed, with the knowledge of the relative entropy, the speed of convergence with Schur pinching is as fast as that of the state-aware protocol, which cannot be achieved with the tomographic approach.
This is made possible because the Schur pinching allows one to entirely skip the estimation process, showcasing the operational capability of exploiting the inherent symmetry in the setting.

We also comment on the relationship to a seemingly related protocol for the universal pure-state entanglement distillation proposed in Ref.~\cite{matsumoto_universal_distortion}.
Their protocol also employs the Schur-Weyl duality based on the permutation symmetry of the input state. However, the idea of utilizing the Schur-Weyl duality is different with each other.  
One reason why this difference happens is that the conditions for the conversion in the quantum thermodynamics and the framework of the bipartite entanglement are opposite in the sense of majorization, and thus the properties of the maximally resourceful state of each theory are different.
Therefore, although our protocol might look simiar to the one in Ref.~\cite{matsumoto_universal_distortion}, the entire construction and how we utilize the symmetry is different to a large extent, reflecting on the different structure of each theory. 
(See Appendix~\ref{App: matsumoto hayashi} for more details.)

Next, we briefly describe the construction of the work extraction protocol in infinite-dimensional systems.
Details of the construction are provided in Appendix~\ref{App sec: Infinite}.
The key idea is to apply the work extraction protocol to the finite-dimensional subspace $\mH_{d}:=\qty{\ket{1},\ldots,\ket{d}}$, where we call $d$ the cutoff dimension.
Since the support $\supp(\rho)$ of the input state $\rho$ might not be included by the subspace $\mH_d$, the work extraction protocol succeeds probabilistically.
If we fix $d$ as a constant, the success probability decays exponentially as the number of subsystems involved in the protocol increases. Thus, we need to take an appropriate sequence $\qty{d_n}_n$ of the cutoff dimensions depending on the number $n$ of copies.
Indeed, there exists a nice choice of such a sequence that the success probability of the work extraction converges to $1$.

The construction of the semiuniversal work extraction protocol in infinite-dimensional systems goes as follows.
As mentioned in the previous discussion, we can consider the thermal operation conditioned by the incoherent measurement because we can implement the action solely for the thermal operation.
Since there are at most a finite number of candidates for the initial state, we can perform the state tomography to identify the given state.
Here, note that the incoherent measurement is not informationally complete, and thus we never obtain the full description of the input state from the measurement outcome.
Nevertheless, we can obtain sufficient information to identify the input state by consuming a constant number of states.
We then apply the work extraction protocol tailored to the identified state, which can be achieved by taking an appropriate series of cutoff dimensions.

\vskip\baselineskip
\begin{center}
{ \bf  DISCUSSION}
\end{center}

We established a universal work extraction protocol, which is independent of the information of the input state and nonetheless can extract the optimal amount of work in the asymptotic limit.
Our result shows that, when we try to extract work from the states, we do not need any information about the input state, but it suffices to know that an i.i.d. state is given, which reduces the enormous cost for obtaining information.
We also extended this universal work extraction protocol to infinite-dimensional systems, providing a semiuniversal optimal work extraction and particularly establishing the optimal extractable work in a state-aware scenario.

Our findings imply that, as long as the i.i.d. states are given, the lack of information about the input state does not affect the performance of the work extraction task in the asymptotic limit.
One possible application of the universal work extraction protocol is to characterize the asymptotic conversion rate between two states with thermal operations under scenarios in which the initial state is unknown.
Another interesting direction is to construct a fully universal work extraction protocol for infinite-dimensional systems that can take a continuous set of states as an input.
The ideas underlying our result may contribute to designing protocols to perform a series of crucial quantum information theoretic tasks in a state-agnostic manner.

\vskip\baselineskip
\begin{center}
{ \bf  DATA AVAILABILITY}
\end{center}
No data sets were generated or analysed during this study.

\let\oldaddcontentsline\addcontentsline% Store \addcontentsline
\renewcommand{\addcontentsline}[3]{}% Make \addcontentsline a no-op

\bibliographystyle{apsrmp4-2}
\bibliography{myref}

\let\addcontentsline\oldaddcontentsline% Restore \addcontentsline

\vskip\baselineskip
\begin{center}
{ \bf  ACKNOWLEDGMENT}
\end{center}

We thank Lorenzo Leone, Yosuke Mitsuhashi, Hayato Arai, Hiroyasu Tajima, Ray Ganardi, and Mizuki Yamaguchi for helpful discussions.
This work is supported by JSPS KAKENHI Grant Number JP23K19028, JP24K16975, JP25K00924, JST, CREST Grant Number JPMJCR23I3, Japan, MEXT KAKENHI Grant-in-Aid for Transformative Research Areas A ``Extreme Universe” Grant Number JP24H00943, and the World-Leading Innovative Graduate Study Program for Advanced Basic Science Course at the University of Tokyo.

\vskip\baselineskip
\begin{center}
{ \bf  AUTHOR CONTRIBUTION}
\end{center}

K.W. and R.T. contributed to every aspect of the research and writing of
this work.

\vskip\baselineskip
\begin{center}
{ \bf COMPETING INTEREST}
\end{center}

The authors declare no competing interests.

\let\oldaddcontentsline\addcontentsline% Store \addcontentsline
\renewcommand{\addcontentsline}[3]{}% Make \addcontentsline a no-op

\let\addcontentsline\oldaddcontentsline% Restore \addcontentsline

%%%%%%%%%%%%%%%%%%%%%%%%%%%%%%%%%%%%%%%%%%
%%%%%%%%%%%%%%%%%%%%%%%%%%%%%%%%%%%%%%%%%%
%%%%%%%%%%%%%%%%%%%%%%%%%%%%%%%%%%%%%%%%%%
%%%%%%%%%%%%%%%%%%%%%%%%%%%%%%%%%%%%%%%%%%

\clearpage
\newgeometry{hmargin=1.2in,vmargin=0.8in}

\widetext
\appendix

\setcounter{thm}{0}
\renewcommand{\thethm}{S.\arabic{thm}}
\setcounter{figure}{0}
\renewcommand{\thefigure}{S.\arabic{figure}}

\begin{center}
{\large \bf Appendices}
\end{center}
\addtocontents{toc}{\protect\setcounter{tocdepth}{1}}
\tableofcontents
% \addcontentsline{toc}{section}{Appendix} % Add the appendix text to the document TOC
% \part{Appendix} % Start the appendix part
% \parttoc % Insert the appendix TOC

%TC:ignore

\section{Technical Lemmas}\label{app sec: lemmas}
\begin{lem}\label{lem: pinching is thermal}
    Let $\qty{\Pi_j}_{j=1}^J$ be the set of projectors satisfying $\sum_j\Pi_j=I$. If all $\Pi_j$ commutes with the Hamiltonian $H$ of the system, i.e., $[\Pi_j,H]=0$, the pinching channel with respect to the projectors $\qty{\Pi_j}_{j=1}^J$ defined as 
    \bal\label{Eq: def of general pinching channel}
    \mP_{\qty{\Pi_j}_j}(\cdot):=\sum_{j=1}^J\Pi_j(\cdot)\Pi_j
    \eal
    is a thermal operation.
\end{lem}
\begin{proof}
    The idea of proof is taken from Ref.~\cite{Gour_role_of}. Eq.~\eqref{Eq: def of general pinching channel} can be rewritten as 
    \bal
    \mP_{\qty{\Pi_j}_j}(\rho)=\frac{1}{J}\sum_{k=1}^J U_k\rho U^\dagger_k,
    \eal
    where, $U_k~(k=1,\ldots,j )$ is the unitary operator defined as
    \bal
    U_k=\sum_{j=1}^J\exp\qty(\frac{2\pi i jk}{J})\Pi_j.
    \eal
    Since all the projectors $\Pi_j$ commute with the Hamiltonian $H$, each $U_k$ is energy conserving. Furthermore, since the set of thermal operations is convex, $\mP_{\qty{\Pi_j}_j}$ is also a thermal operation.
\end{proof}

\begin{lem}\label{Lem: Hoeffding inequality}
    Let $p=(p_1,\ldots, p_d)$ be the probability distribution. Suppose that we sample $n$ times from $p$, and Let $\hat{p}=(\#1/n,\ldots,\# d/n)$ be the empirical distribution. For any $\eta>0$, the probability of obtaining the empirical distribution $\hat{p}$ which is $\eta$-far from the distribution $p$ with respect to the total variational distance is bound as
    \bal
    n\geq \frac{d\ln 2+\ln\frac{1}{\delta}}{2\eta^2}\Rightarrow{\rm Pr}[\|p-\hat{p}\|_1>\eta]\leq \delta
    \eal
\end{lem}

\begin{proof}
    For proof, see Ref.~\cite{Clement_trivial_upper}.
\end{proof}

\begin{lem}[\cite{Bluhm_general_continuity, Bluhm_Continuity_of}]\label{Lem: continuity bound}
    Let $\mH$ be a finite-dimensional Hilbert space, and $H$ be the Hamiltonian. Also, let $\rho_1,\rho_2\in\mD(\mH)$ be two density matrices, and $\tau$ be the thermal state.
    Then, we have
    \bal
    \abs{D(\rho_1\|\tau)-D(\rho_2\|\tau)}\leq \qty(1+\frac{\beta E_{\rm max}+\log Z}{\sqrt{2}})\|\rho_1-\rho_2\|_1
    \eal
    where $E_{\max}$ is the maximum eigenenergy of $H$, and $Z\coloneqq\Tr(e^{-\beta H})$.
\end{lem}
\begin{proof}
    This is the direct consequence of Ref.~\cite{Bluhm_general_continuity} and Ref.~\cite[Corollary 5.9]{Bluhm_Continuity_of}, which states that for any $\rho_1,\rho_2,\tau\in\mD(\mH)$ with $\supp (\rho_1),~\supp(\rho_2) \subset \supp(\tau) $ it holds that
    \bal
    \abs{D(\rho_1\|\tau)-D(\rho_2\|\tau)}\leq \qty(1+\frac{\log\tilde{m}^{-1}}{\sqrt{2}})\|\rho_1-\rho_2\|_1.
    \eal
    Here, $\tilde{m}$ is the minimum nonzero eigenvalue of $\tau$, i.e., $\frac{e^{-\beta E_{\rm max}}}{Z}$. Substituting this, we obtain the inequality.
\end{proof}

\section{State-aware work extraction protocol}\label{App Sec: Brandao}

Before exhibiting the construction of the universal work extraction protocol, we review how the state-aware work extraction protocol for the known input state $\rho^{\otimes n}$ in Ref.~\cite{Brandao_resource_theory_of} works.
We exhibit the construction of the state-aware work extraction protocol in the following way: First, we explain the construction of the optimal work extraction protocol for classical input states.
After that, we explain how to extend this protocol for general quantum states.
There are two ways to extend. One is also explained in the original proof in Ref.~\cite{Brandao_resource_theory_of}, and the other is exhibited in Ref.~\cite{Gour_role_of}, on which the fully universal work extraction protocol is based.

\subsection{State-aware work extraction protocol for classical states}
Suppose that $n$ copies of the classical state $\rho_c$ are given. We aim to extract work by appending $l(\gg n)$ copies of thermal state $\tau$ of the same Hamiltonian and applying the energy-conserving unitary, which is allowed because preparing thermal state is a free operation.
We can assume without loss of generality that the input state $\rho$ and the thermal Gibbs state $\tau$ are simultaneously diagonalized with some energy eigenbasis as 
\bal
\rho_c=\sum_{i=1}^d p_i\ketbra{i}{i}, ~~\tau=\sum_{i=1}^d q_i\ketbra{i}{i}.
\eal
Here, $q_i$ is the Gibbs factor for each energy eigenvalue. 
Utilizing this representation, $\rho^{\otimes n}$ and $\tau^{\otimes l}$ can be written as
\bal
\rho_c^{\otimes n}=\sum_{\Vec{i}\in\qty{1,\ldots,d}^n}p_{\Vec{i}}\ketbra{\Vec{i}}{\Vec{i}},~~\tau^{\otimes l}=\sum_{\Vec{j}\in\qty{1,\ldots,d}^l}q_{\Vec{j}}\ketbra{\Vec{j}}{\Vec{j}}.
\eal
Here, $\ket{\vec{i}}$, $p_{\vec{i}}$, and $q_{\vec{i}}$ represents
\bal
\ket{\vec{i}}:=\ket{i_1}\otimes \cdots\otimes \ket{i_n}, ~p_{\vec{i}}=\prod_{i=1}^n p_i,~q_{\vec{i}}=\prod_{i=1}^n q_i
\eal
For a while, we focus on a fixed energy subspace of $\rho^{\otimes n}\otimes \tau^{\otimes l}$, which is spanned by the vectors
\bal
\ket{\Vec{i}}\otimes \ket{\Vec{j}},~~~\Vec{i}\in{\rm Freq}(n,f),~~\Vec{j}\in{\rm Freq}(l,g).
\eal
Here, ${\rm Freq}(n,f)$ is defined with $f:\qty{1,\ldots,d}\to\mbZ_{\geq 0}$ satisfying $\sum_if(i)=n$ as
\bal\label{eq: def of freq}
{\rm Freq}(n,f)=\qty{\Vec{i}\in\qty{1,\ldots, d}^n~|~ \#\qty{i_m=i, m=1,\ldots, n}=f(i), \forall i\in\qty{1,\ldots, d}}.
\eal
${\rm Freq}(l,g)$ is also defined similarly.

We fix a function $h:\qty{1,\ldots, d}\to\mbZ$ which satisfies $\sum_{i=1}^d h(i)=0.$ This represents the difference between the types before and after applying the energy-conserving unitary. 
When the range of the function $f+g-h$ is nonnegative, ${\rm Freq}(n+l, f+g-h)$ is well-defined as a type class for the strings of length $n+l$. 

To extract some work from $\rho^{\otimes n}\otimes \tau^{\otimes l}$, we apply the energy-conserving unitary which acts on the fixed subspace in the following way.
For any energy subspaces corresponding to the set of strings ${\rm Freq}(n,f)$ and ${\rm Freq}(l,g)$, if ${\rm Freq}(n+l, f+g-h)$ is well-defined, we apply the unitary which map the basis as follows.
\bal\label{app eq: action of work extraction unitary}
\ket{\Vec{i}}\otimes \ket{\Vec{j}}\otimes \ket{0}_X\to \ket{\Vec{k}}\otimes \left|\sum_{i=1}^dE_ih(i)\right\rangle_X.
\eal
Here, $\ket{\Vec{k}}\in\mH^{\otimes (n+l)}$ and $\Vec{k}\in{\rm Freq}(n+l, f+g-h)$. 
To ensure the existence of such a unitary operator, $f,g,h,n,l$ need to satisfy
\bal\label{app eq: condition for injection}
\abs{{\rm Freq}(n,f)}\abs{{\rm Freq}(l,g)}\leq \abs{{\rm Freq}(n+l, f+g-h)}.
\eal
If the energy subspace does not satisfy the unitarity condition or ${\rm Freq}(n+l, f+g-h)$ is not well-defined, we do nothing to that subspace. 
From the definition of ${\rm Freq}(\cdot,\cdot)$, this condition is reduced to
\bal
\frac{n!}{\prod_{i=1}^df(i)!}\frac{l!}{\prod_{i=1}^dg(i)!}\leq \frac{(n+l)!}{\prod_{i=1}^d\qty(f(i)+g(i)-h(i))!}.
\eal
Taking the logarithm and using the Stirling formula, we have
\bal
nH\qty(\frac{\vec{f}}{n})+\mO(\log n)+lH\qty(\frac{\Vec{g}}{l})+\mO(\log l)\leq (n+l)H\qty(\frac{\vec{f}+\vec{g}-\vec{h}}{n+l})+\mO(\log(n+l)).
\eal

Here, $\vec{f}, \vec{g}$ is defined as $\vec{f}=\qty(f(1),\ldots, f(d))^T,\vec{g}=\qty(g(1),\ldots, g(d))^T $, and $H(\vec{p}):=-\sum_ip(i)\log p(i)$ is Shannon entropy. Since we can take $l$ sufficiently larger than the number of the input states $n$, we can take $\eta:=\frac{n}{l}$ arbitrarily small. Dividing both sides by $l$, we get
\bal\label{app eq: entropic condition of the existence of injection}
\eta H\qty(\frac{\vec{f}}{n})+\frac{\mO(\log n)}{l}+H\qty(\frac{\Vec{g}}{l})+\frac{\mO(\log l)}{l}\leq (1+\eta)H\qty(\frac{\vec{f}+\vec{g}-\vec{h}}{n+l})+\frac{\mO(\log(n+l))}{l}.
\eal
The inequality trivially holds when we concentrate on the terms independent of $\eta$.
Focusing on the subdominant terms, this inequality is reduced as
\bal\label{app Eq: condition for the existence of energy conserving unitary in each type class}
\frac{1}{n}\beta \sum h(i)E_i\leq  D\qty(\frac{\vec{f}}{n}\big\|\frac{\vec{g}}{l})-\mO\qty(\frac{\log l}{l}).
\eal
We obtain this inequality by assuming that the energy-conserving unitary with which one can extract work exists. Conversely, when we take any $h$ which satisfies Eq.~\eqref{app Eq: condition for the existence of energy conserving unitary in each type class}, we can find the energy-conserving unitary which realizes the conversion corresponding to the function $h(\cdot)$, because Eq.~\eqref{app Eq: condition for the existence of energy conserving unitary in each type class} is equivalent to the unitarity condition for the terms linear in $\eta$.

In the discussion above, we focus on a fixed subspace corresponding to a single type class. 
When we apply this energy-conserving unitary to the entire state $\rho_c^{\otimes n}\otimes \tau^{\otimes l}$, the vectors $\vec{f},\vec{g}$ run over all possible frequencies.  
However, since the target state is a pure excited state of the work storage, we cannot tailor the function $h(\cdot)$ depending on the type class. 
Thus, after applying the unitary operation determined by the function $h(\cdot)$ to the state $\rho_c^{\otimes n}\otimes \tau^{\otimes l}$ above, we obtain
\bal
(1-\xi)\ketbra{\sum_{i=1}^dE_ih(i)}{\sum_{i=1}^dE_ih(i)}_X+\xi\ketbra{0}{0}_X,
\label{eq:infidelity}
\eal
where $\xi$ is the sum of the trace of the energy blocks to which one cannot apply the energy-conserving unitary. 
To achieve the target error with respect to the fidelity of the final state, we need to tune the appropriate $h(\cdot)$ to make $\xi$ sufficiently small.
Here, note that when we take $n$ and $l$ sufficiently large, because of the law of large numbers, the probability of the frequencies concentrates around the distribution $\rho_c$ and $\tau$. Therefore, it suffices to take $h(\cdot)$ so that the protocol works to the frequencies close to $\rho_c$ and $\tau$. 
Here, note that $l$ needs to be much larger than $n$. One intuitive interpretation is that a sufficiently large thermal bath enables the system to receive the heat needed to charge the work storage. 
However, suppose we take $l=\Omega(n^2)$. In that case, the typical subspace of the system of the input state $\rho_c$ is buried to the fluctuation of thermal states because the standard deviation of the multinomial distribution is $\mO(\sqrt{l})$. Thus the probability concentrates on the range $\pm\frac{1}{\sqrt{l}}$ around $\tau^{\otimes l}$.
Thus, it suffices to take $l=\mO(n^{3/2})$.

When we set $h(\cdot)$ so that the protocol works for the typical subspaces, namely,
\bal\label{Eq: choice of h()}
\frac{1}{n}\beta \sum_i h(i)E_i\leq  D\qty(\rho_c\|\tau)-\mO\qty(\frac{1}{\sqrt{n}})
\eal
holds, the error $\xi$ vanishes asymptotically.

\subsection{State-aware work extraction protocol for general states}\label{app sec: general state state aware}
Now, we review the construction of the state-aware work extraction protocol for the general state, which might not be classical. There are two ways to complete this discussion.

Let us first exhibit the original proof, which is explained in Ref.~\cite[Appendix D]{Brandao_resource_theory_of}. 
Suppose that we are given $n$ copies of $\rho$. 
If the given state has the energetic coherence in it, we apply the pinching channel~\cite{hayashi_optimal_2002, tomamichel_quantum_information_2016} defined as
\bal\label{app eq: pinching}
\mP(\cdot):=\lim_{T\to \infty}\frac{1}{2T}\int_{-T}^{T}\dd t~e^{-iH^{\times n}t}(\cdot)e^{iH^{\times n}t}=\sum_{E_i\in\spec(H^{\times n})}\Pi_{E_i}(\cdot)\Pi_{E_i}.
\eal
After that, our state $\mP(\rho^{\otimes n})$ is block-diagonalized with respect to the energy subspaces because $[\mP(\rho^{\otimes n}), H^{\times n}]=0$. Thus, we can apply an energy-conserving unitary to diagonalize the pinched state.
Indeed, applying the pinching channel does not decrease extractable work, which is referred to as work locking~\cite{Lostaglio_2015}.
After applying the pinching channel, the density matrix of the system is block-diagonalized with respect to the energy eigenbasis, and thus we can apply the work extraction protocol for the classical states.

The idea of construction is described as follows.
The expectation value $\Tr[\mP(\rho^{\otimes n})H^{\times n}]$ of the total energy is written as $n\Tr[\rho H]$, and the variance is at most sublinear in $n$.
Thus, we can concentrate on tailoring the work extraction protocol for the energy subspace with $E\simeq n\Tr[\rho H]$.

This idea that the probability of the pinched state $\mP(\rho^{\otimes n})$ concentrates in the energy subspace around the energy expectation value is rigorously shown in the following way.
Suppose that the Hamiltonian of the single system is diagonalized as $H=\sum_{i=1}^d E_i\ketbra{i}{i}$, and we define a probability distribution as $p_i:=\Tr[\rho\ketbra{i}{i}],~i=1,\ldots, d$. The energy expectation value is represented as $\Tr[\rho H]=\sum_ip_iE_i$.
From the probability distribution $p_i$, we can define a set of $\delta$-(strongly) typical sequences of lenght $n$ as
\bal
\mT_{n.\delta}:=\qty{\vec{i}\in\qty{1,\ldots,d}^n ~|~ \abs{\frac{1}{n}N(i|\vec{i})-p_i }\leq \delta  ~\mbox{if}~p_i\neq 0, ~N(i|\vec{i})=0 ~\mbox{otherwise} }.
\eal
From a property of typical sequence, it holds that for any $\ve>0, \delta>0$ there exists $n_{\ve, \delta}\in\mbN$ such that
\bal
n\geq n_{\ve, \delta} ~\Rightarrow ~{\rm Pr}[\vec{i}\in \mT_{n,\delta}]\geq 1-\ve.
\eal
Moreover, from the Chernoff bound, $n_{\ve, \delta}$ scales as $n_{\ve, \delta}\sim 2^{-n\ve^2}$. 
With respect to the energy eigenvalue, the set of vectors $\ket{\vec{i}}=\ket{i_1}\otimes \cdots\otimes\ket{i_n}$ which corresponds to each elements in $\vec{i}\in\mT_{n,\delta}$ contains the energy eigenstates whose energy is in $[n(E-\delta\Delta E),n(E+\delta\Delta E)]$, where $\Delta E:=E_{\max}-E_{\min}$ is the difference between the maximum and minimum eigenbvalues of the Hamiltonian $H$ of a single system.
Now, we define a projector onto the energy subspace which corresponds to $\mT_{n,\delta}$ as
\bal
\Pi_{\mT_{n,\delta}}:=\sum_{\vec{i}\in \mT_{n,\delta}}\ketbra{i_1}{i_1}\otimes \cdots\otimes \ketbra{i_n}{i_n}.
\eal

From this, for a sufficiently large $n$ it holds that
\bal
\Tr[\mP(\rho^{\otimes n})\Pi_{\mT_{n,\delta}}]&=\Tr[\qty(\sum_{E^{(n)}\in \spec(H^{\times n}) }\Pi_{E^{(n)}}\rho^{\otimes n}\Pi_{E^{(n)}} ) \qty(\sum_{\vec{i}\in\mT_{n,\delta}}\ketbra{i_1}{i_1}\otimes \cdots\otimes \ketbra{i_n}{i_n})]\\
&=\sum_{\vec{i}\in\mT_{n,\delta}}\Tr[\rho^{\otimes n}\ketbra{i_1}{i_1}\otimes \cdots\otimes \ketbra{i_n}{i_n}]\\
&=\sum_{\vec{i}\in\mT_{n,\delta}}p_{i_1}\cdots p_{i_n}={\rm Pr}[\vec{i}\in\mT_{n,\delta}]\geq 1-\ve.
\eal
Thus, to make the fidelity error of the work extraction protocol vanishing in the asymptotic limit it suffices to concentrate on the subspace which corresponds to the projector $\Pi_{\mT_{n,\delta}}$. 

Now, we are in position to discuss how we tailor the work extraction protocol for the arbitrary pinched state $\mP(\rho^{\otimes n})$. As is also discussed above, it suffices to tailor the energy conserving unitary on the typical subspace which acts as in Eq.~\eqref{app eq: action of work extraction unitary}.
The condition for the existence of such an energy-conserving unitary is rephrased as the existence of the injection from the concatenated strings in the typical subspace of $\mP(\rho^{\otimes n})$ and $\tau^{\otimes l}$ to the strings in the typical subspace of the final state. 
Note that the dimension of the typical subspace of $\mP(\rho^{\otimes n})$ is at most $2^{nS(\rho)}$, where $S(\rho):=-\Tr[\rho\log\rho]$ is von Neumann entropy of $\rho$.
Thus, condition for the existence of the energy conserving unitary which functions for the typical energy subspace in this scenarion is written as
\bal
\eta S\qty(\rho)+\frac{\mO(\log n)}{l}+H\qty(\frac{\Vec{g}}{l})+\frac{\mO(\log l)}{l}\leq (1+\eta)H\qty(\frac{\vec{f}+\vec{g}-\vec{h}}{n+l})+\frac{\mO(\log(n+l))}{l}.
\eal
Following the same discussion as above, we can extract work $\sum_{i}h(i)E_i$ if and only if 
\bal\label{app eq: existence of work extraction protocol for general states}
\frac{1}{n}\beta\sum_{i}h(i)E_i\leq D(\rho\|\tau)-\mO\qty(\frac{1}{\sqrt{n}})
\eal
is satisfied. Therefore, for a sufficiently large $n$, the one-shot extractable work from $\rho^{\otimes n}$ in the state-aware scenario behaves as
\bal
\frac{1}{n}\beta W^\ve_{\rm aware}(\rho^{\otimes n})\sim D(\rho\|\tau)-\mO\qty(\frac{1}{\sqrt{n}}),
\eal
which, together with the converse bound $\beta W^\infty_{\rm aware}(\rho)\leq D(\rho\|\tau)$, yields 
\bal
\beta W^\infty_{\rm aware}(\rho)= D(\rho\|\tau).
\eal

We can also obtain the same result by applying the pinching channel to the block and reducing the problem to the classical case, which is also exhibited in Ref.~\cite{Gour_role_of}.
Suppose that $n$ copies of $\rho$, which might not be classical, are given. 
First, we pinch each $k$ copies of $\rho$ for a fixed $k\in\mbN$ and continue this. Then, we obtain $q:=\lfloor\frac{n}{k}\rfloor$ copies of $\mP(\rho^{\otimes k})$, and can be diagonalized by applying an energy-conserving unitary. From this, tehe discussion is reduced to the classical case.
Following the discussion above, the extractable work rate $\frac{1}{k}D(\mP(\rho^{\otimes k})\|\tau^{\otimes k})$ is achievable. We also remark that, since $k$ is an arbitrary positive integer, we can choose $k$ so that the extractable work is maximized. This leads to
\bal
\beta W_{\rm aware}(\rho)&\geq \sup_{k\in\mbN}\frac{1}{k}D(\mP(\rho^{\otimes k})\|\tau^{\otimes k})\\
&\geq \limsup_{k\to\infty}\frac{1}{k}D(\mP(\rho^{\otimes k})\|\tau^{\otimes k}).
\eal
Note that
\bal
D(\rho^{\otimes k}\|\tau^{\otimes k})=D(\rho^{\otimes k}\|\mP(\rho^{\otimes k}))+D(\mP(\rho^{\otimes k})\|\tau^{\otimes k})
\eal
and Hayashi's pinching inequality~\cite{hayashi_optimal_2002} 
\bal
\mP(\rho^{\otimes k})\geq \frac{\rho^{\otimes k}}{\abs{\spec(\tau^{\otimes k})}}
\eal
hold, where $\abs{\spec(\tau^{\otimes k})}$ represents the number of distinct eigenvalues of $\tau^{\otimes k}$. It is straightforward to find the correspondence between the type of a string ${\bf s}\in\qty{1,\ldots, d}^k$ and the energy eigenvalue of the state $\ket{{\bf s}}:=\bigotimes_{i=1}^k\ket{s_i}$. Due to this correspondence, we can show that the number of distinct eigenvalues of $H^{\times k}$ ($=$ number of distinct eigenvalues of $\tau^{\otimes k}$)
is an upper bound by the number of possible type classes of length $k$ with $d$ alphabets, which is at most polynomial in $k$. 

Utilizing the property of the relative entropy, it holds that
\bal
0\leq D(\rho^{\otimes k}\|\mP(\rho^{\otimes k}))&\leq D\qty(\rho^{\otimes k}\|\frac{\rho^{\otimes k}}{\abs{\spec(\tau^{\otimes k})}})\\
&=D(\rho^{\otimes k}\|\rho^{\otimes k})+\log\abs{\spec(\tau^{\otimes k})}\\
&\leq \log {\rm poly}(k).
\eal
Due to this inequality, we obtain
\bal
\beta W_{\rm aware}^\infty (\rho)&\geq \limsup_{k\to\infty}\frac{1}{k}D(\mP(\rho^{\otimes k})\|\tau^{\otimes k})\\
&\geq \limsup_{k\to\infty}\frac{1}{k}\qty(D(\rho^{\otimes k}\|\tau^{\otimes k})-\log {\rm poly}(k))\\
&=\limsup_{k\to\infty}\qty(D(\rho\|\tau)-\frac{1}{k}\log {\rm poly}(k))=D(\rho\|\tau).
\eal

\section{Universal work extraction (Proof of Theorem~\ref{thm:universal})}\label{app sec: universal work extraction}

\subsection{Why the above protocol is state-dependent and how we make this state-agnostic}
The basic idea of our universal work extraction protocol stems from the protocol explained in Appendix~\ref{App Sec: Brandao}. Here, we point out the steps in this protocol that rely on the description of the input state $\rho$ and briefly explain how we resolve them.
See also Figure~\ref{fig: overview of protocol app} for the overview of the protocol.
\begin{figure}
    \centering
    \includegraphics[width=0.8\linewidth]{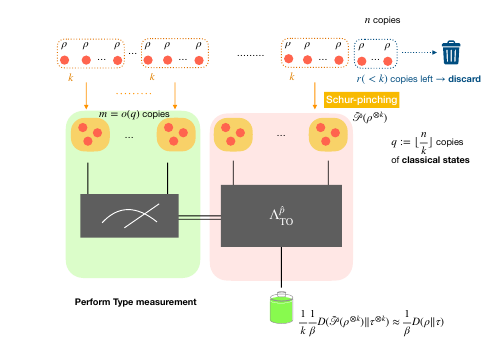}
    \caption{ Overview of the universal work extraction protocol for finite-dimensional systems. First, we apply the channel called Schur pinching to obtain the state diagonalized with a specific energy eigenbasis that also respects the permutation symmetry. After that, we apply an appropriate work extraction protocol, which is conditioned by the type measurement. Since the projector corresponding to the measurement is the projector onto the energy eigenspace, we can realize the same action solely by a thermal operation.}
    \label{fig: overview of protocol app}
\end{figure}

One such step is applying the energy-conserving unitary to diagonalize the pinched state $\mP(\rho^{\otimes k})$. Without a description of the state, we do not know the eigenbasis of $\mP(\rho^{\otimes k})$, preventing us from knowing which unitary to apply.
 A naive way of obtaining a classical state from $\mP(\rho^{\otimes k})$ would be to simply apply a pinching (dephasing) channel with respect to some energy eigenbasis. However, this strategy generally wastes too much free energy, which comes from the coherent (off-diagonal) part of $\mP(\rho^{\otimes k})$.
The difficulty here is to find the right basis with respect to which the pinching does not waste too much free energy \emph{without knowing the description of $\rho$} (and hence $\mP(\rho^{\otimes k})$).
We solve this problem by introducing a new technique, which we call \emph{Schur pinching}, and show that this allows us to retain the optimal rate of work extraction, i.e., the waste by Schur pinching is only sublinear and does not affect the asymptotic work extraction rate.

The other is that we do not know the appropriate $h(\cdot)$, the function that determines the detail of the work extraction protocol for the classical state. Without knowing the value $D(\rho\|\tau)$, it is difficult to determine the function $h(\cdot)$, representing the change in the type class.
We solve this problem by partial tomography approach, where we utilize a sublinear number of copies of input states to estimate $D(\rho\|\tau)$ so that we can determine the right $h(\cdot)$ function. A subtlety here is that thermal operations---the set of free operations allowed for the work extraction process---do not include measurements, making this tomography-based approach appear infeasible. We circumvent this issue by employing \emph{incoherently conditioned thermal operations}---the set of operations introduced in Ref.~\cite{Watanabe_Black_box} containing the thermal operations conditionally chosen depending on the outcomes of incoherent measurements applied to other subsystems. By showing that the tomography-based approach can be formulated as an incoherently conditioned thermal operation with projective measurements, as well as using the known characterization that the latter class of operations is a subset of thermal operations~\cite{Watanabe_Black_box}, we show that the optimal work can be extracted from an unknown classical state by thermal operations.

\subsection{Obtaining a classical state}
Let us address the first issue. Suppose that $n$ copies of the input state $\rho$ are given. As mentioned above and in the main text, we need to obtain the classical state without losing too much resource. To do this, it is helpful to focus on the permutational symmetry of the systems and the states.

Recall that the Schur-Weyl duality induces the decomposition
\bal\label{App eq: schur weyl decomp}
\mH^{\otimes n}=\bigoplus_{\lambda\in Y^{n}_d}\mW_\lambda\otimes \mU_\lambda.
\eal
Here, $Y^n_d$ represents the set of the Young diagrams of depth at most $d$ with $n$ blocks. $\mW_\lambda$ and $\mU_\lambda$ correspond to the Weyl representation of general linear group ${\rm GL}(n,\mbC)$ and the irrep of symmetric group $\mfS_n$ respectively.

For any $\pi\in\mfS_n$, let $V_\pi$ be the unitary matrix representing the permutation of the subsystems.
Any matrix $A_n$ on the Hilbert space $\mH^{\otimes n}$ which is permutationally invariant, i.e., satisfying $V_\pi A_nV_\pi^\dagger=A_n, ~\forall \pi\in\mfS_n$, can be block diagonalized along the direct sum structure in Eq.~\eqref{App eq: schur weyl decomp} as 
\bal
A_n=\bigoplus_{\lambda\in Y^n_d}A_\lambda \otimes I_\lambda,
\eal
which is a direct consequence of Schur's lemma.
Due to this property, we can decompose the input state $\rho^{\otimes n}$ and the global Hamiltonian $H^{\times n}$ as 
\bal\label{App eq: diagonal form of rho and Hamiltonian}
\rho^{\otimes n}&=\bigoplus _{\lambda\in Y^n_d}\rho_\lambda\otimes I_\lambda,\\
H^{\times n}&=\bigoplus _{\lambda\in Y^n_d}H_\lambda\otimes I_\lambda.
\eal
\begin{figure}
    \centering
    \includegraphics[width=0.7\linewidth]{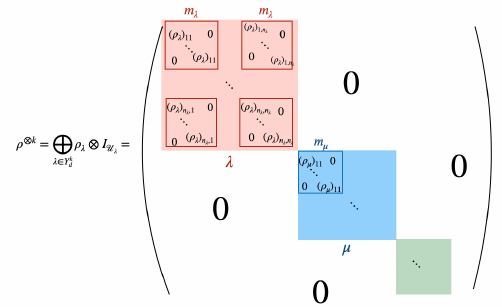}
    \caption{The block diagonal form of the i.i.d. input state.  }
    \label{App fig: structure of input state}
\end{figure}
\begin{figure}
    \centering
    \includegraphics[width=0.8\linewidth]{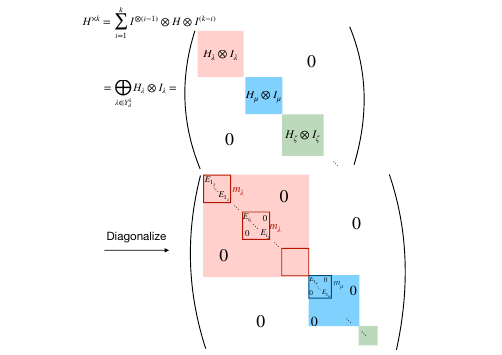}
    \caption{The block diagonal form of the Hamiltonian. Here, since $H_\lambda$ is Hermitian for any $\lambda\in Y^n_d$, we can take an appropriate orthogonal basis within each block to diagonalize the Hamiltonian.}
    \label{App fig: Schur decomp of Hamiltonian}
\end{figure}

Now, we are ready to obtain a diagonal state without applying a unitary.
To this end, we focus on the eigenspace of $H_\lambda\otimes I_\lambda$. 
The projector onto the eigenspace of the matrix $H_\lambda\otimes I_\lambda$ corresponding to the eigenvalue $E_{i_\lambda}$ is denoted by $\Pi_{E_{i_\lambda}}$.
Importantly, we know the description of $\Pi_{E_{i_\lambda}}$---and thus the eigenbasis of $H_\lambda\otimes I_\lambda$---because the basis that admits the decomposition \eqref{App eq: diagonal form of rho and Hamiltonian} (known as Schur basis), as well as the description of the Hamiltonian, is known and independent of $\rho$. This allows us to write $H^{\times k}$ as blocks of identity matrices with size $m_\lambda\coloneqq \dim \mU_\lambda$, each of which is multiplied by the energy eigenvalues $E_{i_\lambda}$ (Figure~\ref{App fig: Schur decomp of Hamiltonian}).

Here, we remark that the union set of eigenvalues of all $H_\lambda\otimes I_\lambda$ coincides with that of $H^{\times n}$ itself. 
Also, note that each eigenspace of $H_\lambda\otimes I_\lambda$ is the subspace of the energy eigenspace of the global Hamiltonian $H^{\times n}$ corresponding to the same energy eigenvalue.
In this sense, this projector is fine-grained compared to that of the original Hamiltonian.
We introduce \emph{Schur pinching channel} as
\bal
\tilde{\mP}(\cdot)=\sum_\lambda\sum_{E_{i_\lambda}}\Pi_{E_{i_\lambda}}(\cdot)\Pi_{E_{i_\lambda}}.
\eal
Here, note that, due to Lemma~\ref{lem: pinching is thermal}, Schur pinching defined above is a thermal operation.
After applying this channel, whose action is exhibited in Figure~\ref{App fig: Action of Schur pinching}, we obtain the classical state.
\begin{figure}
    \centering
    \includegraphics[width=0.8\linewidth]{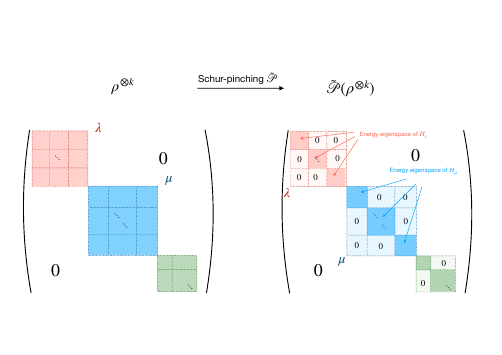}
    \caption{Action of Schur pinching channel. Each block of $\rho^{\otimes k}$ is diagonalized due to the decomposition in Eq.~\eqref{App eq: diagonal form of rho and Hamiltonian}, but the matrix form of $\rho^{\otimes k}$ itself is not diagonalized. Schur pinching vanishes all the off-diagonal blocks; thus, we obtain a diagonal state.}
    \label{App fig: Action of Schur pinching}
\end{figure}

Let us see an example of how we obtain Schur pinching.
For simplicity, we assume that the single system is two-dimensional and thus a qubit system, and we have three copies of this system.
We assume that the Hamiltonian is given as $H=E\ketbra{1}{1}+0\ketbra{0}{0}$.
The Hilbert space of the composite system can be decomposed due to the Schur-Weyl duality as
\bal
\mH_2^{\otimes 3}&=\bigoplus_{\lambda\in Y^3_2}\mW_\lambda\otimes \mU_{\lambda}\\
&=\qty(\mW_{\scriptsize{\yng(3)}}\otimes \mU_{\scriptsize{\yng(3)}})\oplus \qty(\mW_{\scriptsize{\yng(2,1)}}\otimes \mU_{\scriptsize{\yng(2,1)}})
\eal
It is known that the dimensions of these irreps are given as
\bal
\dim \mW_{\scriptsize{\yng(3)}}=4,~\dim \mU_{\scriptsize{\yng(3)}}=1,~\dim \mW_{\scriptsize{\yng(2,1)}}=2,~\dim \mU_{\scriptsize{\yng(2,1)}}=2.
\eal
The irrep $\mW_{\scriptsize{\yng(3)}}\otimes\mU_{\scriptsize{\yng(3)}} $ is spanned by 
\bal
\ket{v_1}&=\ket{000}\\
\ket{v_2}&=\frac{\ket{001}+\ket{010}+\ket{100}}{\sqrt{3}} ,\\
\ket{v_3}&=\frac{\ket{011}+\ket{101}+\ket{110}}{\sqrt{3}},\\
\ket{v_4}&=\ket{111},
\eal
and the irrep $\mW_{\scriptsize{\yng(2,1)}}\otimes\mU_{\scriptsize{\yng(2,1)}} $ is spanned by 
\bal
\ket{u_1}&=\frac{\ket{100}-\ket{010}}{\sqrt{2}},\\
\ket{u_2}&=\frac{\ket{101}-\ket{011}}{\sqrt{2}},\\
\ket{u_3}&=\frac{2\ket{001}-\ket{010}-\ket{100} }{\sqrt{6}},\\
\ket{u_4}&=\frac{2\ket{110}-\ket{101}-\ket{011} }{\sqrt{6}}.
\eal
Noting that the type of each basis corresponds to the energy eigenvalue,  Schur pinching for this Hamiltonian is defined with the following projectors.
\bal
\Pi_{\scriptsize{\yng(3)},0}&=\ketbra{v_1}{v_1}, \\
\Pi_{\scriptsize{\yng(3)},E}&=\ketbra{v_2}{v_2},\\
\Pi_{\scriptsize{\yng(3)},2E}&=\ketbra{v_3}{v_3},\\
\Pi_{\scriptsize{\yng(3)},3E}&=\ketbra{v_4}{v_4},\\
\Pi_{\scriptsize{\yng(2,1)},E}&=\ketbra{u_1}{u_1}+\ketbra{u_3}{u_3},\\
\Pi_{\scriptsize{\yng(2,1)},2E}&=\ketbra{u_2}{u_2}+\ketbra{u_4}{u_4}.
\eal

We now show that the Schur pinching only decreases the free energy by the amount that vanishes in the asymptotic limit. To this end, let us evaluate the number of projectors involved in the definition of the Schur pinching. 
The number $\abs{Y^n_d}$ of the Young diagrams of depth at most $d$ composed of $n$ boxes can be bound as
\bal
\abs{Y^n_d}\leq (n+1)^{d-1}.
\eal
This can be easily seen in the same way as the bound of the number of the type classes of the strings with $d$ alphabets and length $n$, which is the following.
From each Young diagram $\lambda \in Y^n_d$, we can make the vector
$\vec{\lambda}=(\lambda_1,\ldots ,\lambda_d)$, where $\lambda_1,\ldots ,\lambda_d$ indicates the number of the boxes in each row. Since $\lambda$ contains $n$ box in total, it holds that $\sum_{i=1}^d\lambda_i=n$. It suffices to find a bound on the number of possible vectors $\vec{\lambda}$. If we specify $d-1$ entries, the $d$th entry is determined automatically. Thus, it is sufficient to concentrate on $d-1$ entries. Since the possible value for each element is $0,1,\ldots, n$, we obtain the bound above.

The evaluation of the number of distinct eigenvalues also goes similarly. Since the energy eigenvalue of the global Hamiltonian $H^{\times n}$ also corresponds to the type classes of the strings ${\bf s}\in\qty{1,\ldots, d}^n$ and thus $\abs{\spec(H^{\times n})}$ can be bound as
\bal
\abs{\spec(H^{\times n})}\leq (n+1)^{d-1}.
\eal
Since $H^{\times n}$ is block-diagonalized as the form in Eq.~\eqref{App eq: diagonal form of rho and Hamiltonian}, $\abs{\spec(H_\lambda\otimes I_\lambda)}\leq (n+1)^{d-1}$ also holds for any Young diagrams $\lambda$.
These observations conclude that the number of projectors involved in the Schur pinching is at most polynomial in $n$, due to the following calculation.
\bal
\sum_{\lambda\in Y^n_d}\sum_{E_{i_\lambda}}1 &= \abs{Y^n_d}\abs{\spec(H_\lambda\otimes I_\lambda)}\\
&\leq (n+1)^{2(d-1)}.
\label{eq:Schur pinching number of projectors}
\eal

We remark on how we employ the Schur pinching in our universal work extraction protocol.
Given $n$ copies of $\rho$, we apply the Schur pinching to $k$ copies. Repeating this, we obtain $q:=\lfloor\frac{n}{k}\rfloor$ copies of the classical state $\tilde{\mP}(\rho^{\otimes k})$.  
We then apply the protocol we describe in the next subsection that universally extracts optimal work from an arbitrary unknown classical state, which achieve the rate $\frac{1}{k}D(\tilde{\mP}(\rho^{\otimes k})\|\tau^{\otimes k})$ in the asymptotic limit.
Here, note that $\tilde{\mP}(\rho^{\otimes k})$ is the state over $k$ systems, and that is why the coefficient $\frac{1}{k}$ is needed.

Hayashi's pinching inequality~\cite{hayashi_optimal_2002}, together with \eqref{eq:Schur pinching number of projectors}, implies that 
\bal
 \tilde \mP(\rho^{\otimes k})\geq \frac{\rho^{\otimes k}}{(k+1)^{2(d-1)}}
 \label{eq:pinching inequality Schur pinching}.
\eal

Therefore, it holds that
\bal\label{app Eq: bound for the schur pinching}
\abs{\frac{1}{k}D(\tilde{\mP}(\rho^{\otimes k})\|\tau^{\otimes k})-D(\rho\|\tau)} &=\frac{1}{k}D(\rho^{\otimes k}\|\tilde\mP(\rho^{\otimes k}))\\
&\leq \frac{1}{k}\log\left[ (k+1)^{2(d-1)}\right].
\eal
Because of this, if we take $k$ sufficiently large, the work extraction rate becomes arbitrarily close to $D(\rho\|\tau)$.
Therefore, our universal work extraction protocol will be completed, given the construction of the universal work extraction protocol for classical input states, which we describe in the following.

\subsection{Universal work extraction protocol}
 The obstacle for applying the protocol in Ref.~\cite{Brandao_resource_theory_of} mentioned in Appendix~\ref{App Sec: Brandao} is that, we do not know the relative entropy $D(\rho_c\|\tau)$ of the input state $\rho_c$ and thus we do not know which $h(\cdot)$ to
choose. 
Here, recall that $h:\qty{1,\ldots,d}\to\mbZ$ is the function that determines the action of the energy-conserving unitary applied to the systems.

The strategy to overcome this is straightforward: We perform the type measurement on a sublinear number of systems and estimate the state. 
The decrease of the system available to the work extraction protocol is negligible in the asymptotic limit.
This enables one to apply the work extraction protocol according to the measurement result.

The problem with this approach is that performing a measurement is not a thermal operation. Indeed, this procedure is one of the operations in a class of operations called conditioned thermal operation introduced in Ref.~\cite{Narasimhachar_CTO}, which is generally strictly larger than thermal operations.
Therefore, the thermal operations conditioned by the measurement outcome cannot be done with thermal operations.
Nevertheless, as is seen in the subsequent discussion, a restricted class of conditioned thermal operations is sufficient for our purpose, and this class of conditioned thermal operations can be implemented solely by thermal operations.

To show this, we review a class of measurement called incoherent measurement and a class of operations called incoherently conditioned thermal operations, introduced in Ref.~\cite{Watanabe_Black_box}.
\begin{defn}
Let $\mH$ be a Hilbert space, $H$ be the Hamiltonian of the system and $\qty{M_a}_{a}$ be a complete set of the POVM elements, which implies that the family of operators $\qty{M_a}_{a}$ satisfies $M_a\geq 0,~\forall a$ and $\sum_aM_a=I$. $\qty{M_a}_a$ is said to be a \emph{incoherent measurement} if the POVM elements $M_a$ satisfies $\mP(M_a)=M_a$ for any $a$, where $\mP$ is the pinching channel with respect to the Hamiltonian $H$ defined in Eq.~\eqref{app eq: pinching}.

        Let $\mH_A,\mH_B,\mH_C$ be Hilbert spaces, and $\mE:\mD(\mH_A\otimes \mH_B)\to \mD(\mH_C)$ be a CPTP map. $\mE$ is called an \emph{incoherently conditioned thermal operation} when $\mE$ can be decomposed as follows.
    \bal
    \mE=\sum_a\mE^{\rm TO}_a\circ \Lambda^{\rm meas}_a.
    \eal
    Here, $\mE_a^{\rm TO}:\mD(\mH_B)\to\mD(\mH_C),~a=1,2,\ldots,m$ are thermal operations and 
    \bal
    \Lambda^{\rm meas}_a(\rho_{AB}):=\Tr_A\qty[(M^{\rm incoh}_a\otimes I_{B})\rho_{AB}],~a=1,\ldots ,m
    \eal
    be the instruments, where $\qty{M^{\rm incoh}_a}_a$ is the family of the POVM elements of incoherent measurement on $\mD(\mH_A)$.
     We denote the set of incoherently conditioned thermal operations as ${\rm{ICTO}}(A; B\to C)$.
    Furthermore, when the measurements are restricted to the incoherent projective measurements, we say that the operation is a thermal operation $+$ incoherent projective measurement. We denote the set of thermal operations $+$ incoherent projective measurements as ${\rm{ICPTO}}(A;B\to C)$
\end{defn}
Note that the post-measured state is discarded. 
The physical interpretation of the incoherent measurement is that the measurement cannot detect any information about the time evolution.
The class of thermal operations is a subclass of incoherently conditioned thermal operations.
Ref.~\cite[Proposition S.23]{Watanabe_Black_box} showed that the incoherently conditioned thermal operations whose measurement is projective are thermal operations.

\begin{lem}[\cite{Watanabe_Black_box}]\label{app lem: TO=ICPTO}
Let $A, B$ be the input system, and $C$ be the output system. It holds that
\bal
    \rm{TO}(AB\to C)=\rm{ICPTO}(A;B\to C).
    \eal
\end{lem}

The intuitive explanation is that the energy-conserving unitary conditioned by the incoherent projective measurement can be replaced by a controlled energy-conserving unitary.
Lemma~\ref{app lem: TO=ICPTO} guarantees that we can perform the type measurement of the energy eigenbasis on the subsystems because it can also be done with the thermal operations.

In the subsequent discussion, we exhibit the construction of the universal work extraction from the given $n$ copies of states $\rho^{\otimes n}$.
The overview of the protocol after pinching is as follows.
\begin{enumerate}
    \item Perform the type measurement with a sublinear number of systems. 
    \item According to the measurement result, calculate the estimated relative entropy.
    \item Apply the protocol in Sec.~\ref{App Sec: Brandao}.
\end{enumerate}
Let $\qty{k_n}_{n\in\mbN},~\qty{m_n}_{n\in\mbN}$ be two sequences of the positive integers, and $\qty{\ve_n}_{n\in\mbN},~\qty{\delta_n}_{n\in\mbN}$ be two sequences of positive numbers.
$\qty{k_n}_{n\in\mbN}$ describes the increase of the number of states in a pinched block, and $\qty{m_n}_{n\in\mbN}$ indicates the number of measured states out of $q_n:=\lfloor\frac{n}{k_n}\rfloor$ copies of $\tilde{\mP}(\rho^{\otimes k})$. 
To achieve the optimal work extraction rate $D(\rho\|\tau)$, $\qty{k_n}$ and $\qty{m_n}$ must satisfy at least $k_n\to \infty,~\frac{k_n}{n}\to 0,~ \frac{m_n}{q_n}\to 0$. The first condition is to make $\frac{1}{k}D(\tilde{\mP}(\rho^{\otimes k})\|\tau^{\otimes k})$ arbitrarily close to $D(\rho\|\tau)$.
The second condition is to make $q_n\to\infty$. The last condition is necessary so that $m_n$ is negligible.
Moreover, $\qty{\ve_n}$ and $\qty{\delta_n}$ respectively represent the error of the fidelity of the final state and the error in the work extraction rate, which should satisfy $\ve_n\to 0$ and $\delta_n\to 0$.
Our goal is to find the appropriate $\qty{k_n}$ and $\qty{m_n}$ which achieves the vanishing errors $\qty{\ve_n},\qty{\delta_n}$.

Here, note that when the protocol succeeds probabilistically, the probability of failure affects the fidelity error of the work extraction protocol. 
To see this, let us denote the ideal CPTP map as $\Lambda$, and the probability of failure as $p_e$.
In this case, the operation applied to the state is written as $(1-p_e)\Lambda+p_e\Xi$, where $\Xi$ is some CPTP map.
For an arbitrary input state $\rho$, the fidelity error of this probabilistic protocol is bound as
\bal
1- F\qty(((1-p_e)\Lambda+p_e\Xi)(\rho), \Lambda(\rho))\leq 1-(1-p_e)F(\Lambda(\rho),\Lambda(\rho))=p_e.
\eal
Therefore, we need to make the probability of failure in the probabilistic part of the work extraction small enough to achieve the target fidelity error.

Thanks to the Schur pinching, the state $\tilde{\mP}(\rho^{\otimes k})$ and the $k$ copies of the thermal state $\tau$ are diagonalized by a fixed energy eigenbasis. We denote this eigenbasis as $\qty{\ket{i}}_{i=1}^{d^k}$, and we assume that the two states are decomposed as
\bal
\tilde{\mP}(\rho^{\otimes k})=\sum_{i=1}^{d^k}(p^k)_i\ketbra{i}{i},~~\tau^{\otimes k}=\sum_{i=1}^{d^k}(t^k)_i\ketbra{i}{i}.
\eal
For simplicity, we denote the probability distribution $((p^k)_i)_{i=1}^{d^k}$ and $((t^k)_i)_{i=1}^{d^k}$ as $p^k$ and $t^k$ respectively.

Let us fix $n\in\mbN$. 
We first evaluate the number of the measured systems $m_n$, which is sufficient to estimate the probability distribution $p^k$ with error $r_n$ with respect to the total variational distance with probability larger than $\frac{\ve_n}{2}$.
Let us denote the empirical distribution of the type measurement of the string with length $m_n$ composed of the alphabets $\qty{1,\ldots, d^k_n}$ as $\hat{p}^k$. 
Due to Lemma~\ref{Lem: Hoeffding inequality}, we have
\bal\label{Eq: sample comlexity of tomography}
    m_n\geq \frac{d^{k_n}\ln 2+\ln\frac{2}{\ve_n}}{2 r_n^2}\Rightarrow{\rm Pr}[\|p^k-\hat{p}^k\|_1> r_n]\leq \frac{\ve_n}{2}.
\eal
By the continuity bound of the relative entropy referred to in Lemma~\ref{Lem: continuity bound}, it holds that
\bal\label{App Eq: continuity bound and estimation}
\|p^k-\hat{p}^k\|_1\leq r_n~\Rightarrow~\abs{D(p^k\|t^k)-D(\hat{p}^k\|t^k)}\leq \qty(1+k_n\frac{\beta E_{\rm max}+\log Z}{\sqrt{2}})r_n.
\eal
where $E_{\max}$ is the maximum eigenenergy of the Hamiltonian $H$ of the single system, and $Z=\Tr(e^{-\beta H})$.
 \begin{figure}
       \centering
       \includegraphics[width=0.8\linewidth]{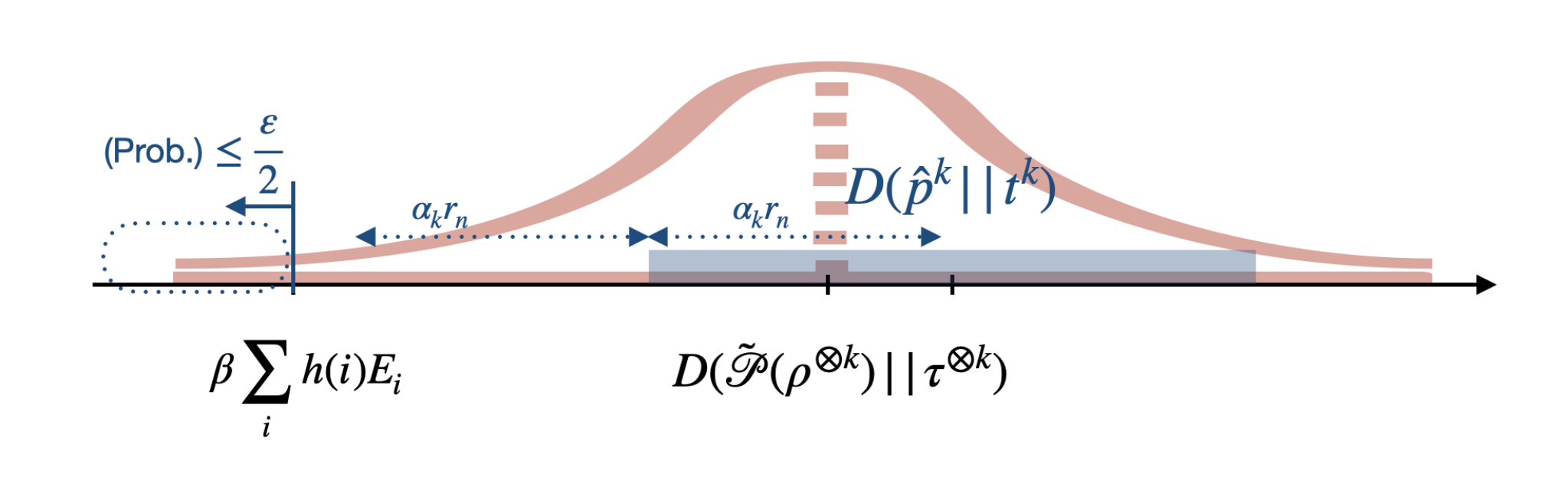}
       \caption{The schematic figure of the work extraction protocol in the state-agnostic scenario. The red curve represents the probability of the relative entropy of each type subspace. Due to the typicality, the peak of the red curve is at the relative entropy $D(\tilde{\mP}(\rho^{\otimes k})\|\tau^{\otimes k})=D(p^k\|t^k)$ of the pinched input state. However, because we do not have direct access to this distribution, the accessible information is that the relative entropy of the pinched input state $D(\tilde{\mP}(\rho^{\otimes k})\|\tau^{\otimes k})$ is at most $\qty(1+k_n\frac{\beta E_{\rm max}+\log Z}{\sqrt{2}})r_n\qty(:=\alpha_k r_n)$ far from the estimated relative entropy $D(\hat{p}^k\|t^k)$ with probability larger than $1-\ve/2$. As is also mentioned in the discussion of Eq.~\eqref{eq:infidelity}, we need to make the error from the atypical subspace to which the energy-conserving unitary cannot be applied. To ensure a sufficiently small error on fidelity, it is sufficient to take the function $h(\cdot)$ so that $\beta\sum_ih(i)E_i$ per subsystems involved in the work extraction is equal to smaller than $D(\hat{p}^k\|t^k)-2\alpha_k r_n$.  }
       \label{fig:choice of protocol}
\end{figure}
To specify the energy-conserving unitary, we need to find $h(\cdot)$ that satisfies Eq.~\eqref{Eq: choice of h()}. Considering the estimation error of the relative entropy, it suffices to design $h(\cdot)$ so that 
\bal\label{app eq: guaranteed performence after pinch}
\frac{1}{q_n-m_n}\beta \sum_i h(i)E_i\leq  D\qty(\hat{p}^k\|t^k)-2\qty(1+k_n\frac{\beta E_{\rm max}+\log Z}{\sqrt{2}})r_n-\mO\qty(\frac{1}{\sqrt{q_n-m_n}}).
\eal
holds (See Figure~\ref{fig:choice of protocol}).
Now, we derive the conditions of $m_n$ and $k_n$ so that the target errors $\ve_n$ and $\delta_n$ are achieved.
There are two factors of the fidelity error, the probability of failing the type measurement and the effect of atypical block in Eq.~\eqref{eq:infidelity}. If we take $q_n$ and $m_n$ such that $q_n\gg m_n$ holds and $m_n$ satisfies Eq.~\eqref{Eq: sample comlexity of tomography}, the total trace of atypical space is less than $\frac{\ve_n}{2}$.

Now, we concentrate on the error on the work extraction rate $\delta_n$.
There are three factors that affect the error on the work extraction rate: The decrease of the number used to extract work, the effect of Schur pinching, and estimation error.
First, we evaluate the ratio of the number $k_n(q_n-m_n)$ of the states used to apply the work extraction per the total number $n$ of input states. We denote the number of discarded systems before the Schur pinching as $s_n:=n-k_n\lfloor\frac{n}{k_n}\rfloor$. If we assume that $\frac{k_n}{n}\to 0$ and $\frac{m_n}{q_n}\to 0$ hold in the limit $n\to \infty$, we have
\bal
1\geq\frac{k_n(q_n-m_n)}{n}=\frac{k_n(q_n-m_n)}{k_n q_n+s_n}\geq \frac{k_n(q_n-m_n)}{k_n(q_n+1)}\geq \frac{q_n-m_n}{q_n+1}\xrightarrow{n\to\infty} 1.
\eal
From this, we can see that the effect of the measured systems $m_n$ and the discarded systems $s_n$ can be negligible.

Eq.~\eqref{app Eq: bound for the schur pinching} ensures that, as long as we take $\qty{k_n}_{n\in\mbN}$ so that $k_n\to \infty$ holds in the asymptotic limit, the effect of pinching is also negligible in the limit $n\to\infty$.
Thus, we need to check that we can make the estimation error arbitrarily small. We set the target error only for this error as $\delta'_n$.
Due to Eq.~\eqref{App Eq: continuity bound and estimation}, it is sufficient to take $r_n$ satisfying
\bal
3\qty(1+k_n\frac{\beta E_{\rm max}+\log Z}{\sqrt{2}})r_n\leq \delta'_n.
\eal
Substituting this to Eq.\eqref{Eq: sample comlexity of tomography}, we have
\bal
m_n\geq 9\frac{d^{k_n}\ln 2+\ln\frac{2}{\ve_n}}{2{\delta'}_n^2}\qty(1+k_n\frac{\beta E_{\rm max}+\ln Z}{\sqrt{2}})^2.
\eal
If we take the sequences $\qty{k_n}_{n\in\mbN},~\qty{\ve_n}_{n\in\mbN},~\qty{\delta_n}_{n\in\mbN}$ as 
\bal\label{app eq: choice of parameters in finite}
k_n&=\lfloor\frac{\ln n}{3\ln d}\rfloor,\\
\ve_n&=e^{-n^{1/3}},\\
\delta'_n&=\frac{1}{n^{1/6}},
\eal
$m_n$ scales as $m_n=\mO\qty(n^{2/3}(\log n)^2)$. 
The following confirms that this choice of parameters satisfies the assumptions of the parameters imposed so far.
\bal
q_n=\frac{n}{k_n}&\sim \frac{ n}{\log n}\xrightarrow{n\to\infty}\infty,\\
\frac{m_n}{q_n}&\sim\frac{n^{2/3}(\log n)^2}{\frac{n}{\log n}}=\frac{(\log n)^3}{n^{1/3}}\xrightarrow{n\to\infty}0.
\eal

Combining the above discussions, we can conclude that if we tailor the thermal operations characterized by the parameters in Eq.~\eqref{app eq: choice of parameters in finite}, the protocol achieves the work extraction task $D(\rho\|\tau)$ from any input states $\qty{\rho^{\otimes n}}_{n\in\mbN}$.
This concludes the proof of Theorem~\ref{thm:universal}: the universal work extraction for finite-dimensional systems.

\subsection{ Alternative universal protocols}\label{app sec: measure and prepare}
We present two other approaches to obtaining the universal work extraction protocol, which offer different perspectives and may be of independent interest.

\subsubsection{Measure-and-prepare strategy}
One approach utilizes the idea of the measure-and-prepare channel.
Although the scope of the protocol is slightly narrower than the protocol above, this protocol is insightful.
Before the construction, we state the restrictions of this protocol.
One is that the protocol exhibited below cannot work for some states. Nevertheless, the volume of such states is measure-zero, and thus the possibility of accidentally having such input states is negligible.
The other is that the class of free operations of this protocol is chosen as the closure of thermal operations.
This is, however, not a major drawback either, considering that taking the closure for thermal operations is a standard treatment in quantum thermodynamics~\cite{Gour_role_of}.

Since Schur pinching allows us to diagonalize any input i.i.d. states with a fixed energy eigenbasis, we can concentrate on the classical state without loss of generality.
In the subsequent discussion, we employ another equivalent definition of extractable work utilizing another system called \emph{battery}. In Ref.~\cite{Gour_role_of}, they consider the battery system with two energy levels. We extend this definition to the system with multiple energy levels.

Let us consider the $d_X$-dimensional Hilbert space $\mH_X$ with the Hamiltonian which is given as
\bal
H_X=\sum_{w=1}^{d_X} E_w \ketbra{w}{w}.
\eal
Given a state $\rho$, suppose that the conversion $\rho\to \ketbra{w}{w}$ is possible with the free operation, we define the extracted work as
\bal
\beta W(\rho)=\log\frac{1}{\bra{w}\tau_X\ket{w}}.
\eal
Compared to the work storage, the advantage of the battery system is that we do not need to append the ground state in advance.
The definition of extractable work with a work storage and that with a battery system are equivalent in the following sense.
\begin{lem}\label{lem: work storage=work battery}
   If the work $W$ is extractable in the sense of the work storage, i.e., the conversion $\rho\otimes \ketbra{0}{0}_W\to \rho'\otimes \ketbra{W}{W}_W$ is possible by the thermal operations, the work $W$ is extractable in the sense of the work battery, i.e., $\rho\otimes \ketbra{0}{0}_W\to\ketbra{w}{w}_X\otimes \ketbra{0}{0}_W$ is also possible with thermal operations. Here, $\ket{w}$ is the energy eigenstate of $H_X$ with $W=-\frac{1}{\beta}\log \bra{w}H_X\ket{w}$. Conversely, if we can extract work $W$ from the initial state $\rho$ in the sense of the work battery, we can extract work from the initial state $\rho$ in the sense of the work storage.
\end{lem}

\begin{proof}
    In Ref.~\cite{horodecki_fundamental_2013}, it is shown that the extractable work and the work formation with thermal operations of the state $\ketbra{w}{w}$ in the sense of the work storage is given by 
    \bal
    \beta W_{\rm ext}(\ketbra{w}{w}_X)&=D_{\min}(\ketbra{w}{w}\|\tau_X)\\
    \beta W_{\rm form}(\ketbra{w}{w}_X)&=D_{\rm max}(\ketbra{w}{w}\|\tau_X).
    \eal
    Here, $W_{\rm ext}$ is the exact extractable work, and $W_{\rm form}$ is the exact work formation defined as
    \bal
    W_{\rm ext}(\rho)&=\max\qty{W~|~\exists \Lambda\in{\rm TO},~{\rm s.t.}~\Lambda(\rho\otimes\ketbra{0}{0}_W)=\ketbra{W}{W}_W},\\
    W_{\rm form}(\rho)&=\min\qty{W~|~\exists \Lambda\in{\rm TO},~{\rm s.t.}~\Lambda(\ketbra{W}{W}_W)=\rho\otimes\ketbra{0}{0}_W},
    \eal
    and $D_{\rm min}$ and $D_{\rm max}$ represents the quantum min- and max-relative entropy defined as
    \bal
    D_{\rm min}(\rho\|\sigma)&=-\log \Tr[\Pi_{\rho}\sigma],\\
    D_{\rm max}(\rho\|\tau)&=\log \inf\qty{\lambda~|~\rho\leq \lambda \sigma}.
    \eal
    From the definition, it holds that
    \bal
    D_{\min}(\ketbra{w}{w}_X\|\tau_X)=D_{\max}(\ketbra{w}{w}_X\|\tau_X)=\log\frac{1}{\bra{w}\tau_X\ket{w}}.
    \eal
    
    From this observation, we reach the statement as follows.
    Suppose that the conversion $\rho\otimes \ketbra{0}{0}_W\to \rho'\otimes \ketbra{W}{W}_W$ is possible. 
    Then, we append the thermal state $\tau_X$ of the work battery. Since the work formation of the state $\ketbra{w}{w}$ is expressed as $\beta W_{\rm form}(\ketbra{w}{w})=D_{\min}(\ketbra{w}{w}_X\|\tau_X)=\frac{1}{\beta}\log \bra{w}H_X\ket{w}=\beta W$, we can create $\ketbra{w}{w}_X$ with a thermal operation with the state $\ketbra{W}{W}_W$ of the work storage. 
    
    We can show the opposite statement similarly. Suppose that the conversion $\rho\to\ketbra{w}{w}_X$ is possible by applying a thermal operation. Note that the extractable work of the state $\ketbra{w}{w}$ with respect to the work storage is $\beta W_{\rm ext}(\ketbra{w}{w}_X)=D_{\min}(\ketbra{w}{w}\|\tau_X)=\frac{1}{\beta}\log \bra{w}H_X\ket{w}=\beta W$. From this, we conclude that work $W$ can be extracted from the input $\rho$ in the sense of the work storage as follows. First, we append the ground state $\ketbra{0}{0}_W$ of the work storage to the input state $\rho$. From the assumption, the conversion $\rho\otimes \ketbra{0}{0}_W\to \ketbra{w}{w}_X\otimes \ketbra{0}{0}_W$ is possible by a thermal operation. Finally, we apply the work extraction protocol for the state $\ketbra{w}{w}_X$ of the work battery and obtain $\ketbra{W}{W}_W$. Combining these steps, we can show that the conversion $\rho\otimes \ketbra{0}{0}_W\to \ketbra{W}{W}_W$ is possible.
\end{proof}

Suppose we are given $n$ copies of quantum state $\rho$. Because of the Schur pinching, we obtain the classical state, which is diagonalized with a known basis.
Combined with Eq.~\eqref{app Eq: bound for the schur pinching}, we can assume that the input state is classical without loss of generality.

In Ref.~\cite{Shiraishi_2020_two_constructive, horodecki_fundamental_2013}, they show that the classical Gibbs preserving operation can be enumerated by the thermal operation with an infinitely large bath. 
Thus it suffices to consider the Gibbs preserving operation.
Here, the classical Gibbs-preserving operation is defined as follows.
\begin{defn}
    Suppose the system has $D$ possible states, whose energies are $E_1,\ldots, E_D$. A probability transition matrix $T=\qty(T_{ij})_{i,j=1}^D$ is called a Gibbs preserving operations if 
    \bal
    t=Tt
    \eal
    holds, where $t$ is the probability distribution of the thermal state.
\end{defn}

We denote the set of all probability distributions with $D$ entries as $\mbP_D$.
We divide the set of all possible strings into groups that only contain similar types as follows.
\begin{defn}
    Fix a natural number $M$, and a vector of nonnegative integers $\Vec{l}=(l_1,\ldots, l_D)$ which satisfies $M=\sum_{i=1}^D l_i$. A block $\mB(\Vec{l})$ corresponding to $\Vec{l}$ is defined as
    \bal
    \mB(\vec{l})=\qty{\Vec{p}\in\mbP_D~|~\frac{\Vec{l}}{M}=\arg\min_{\frac{\Vec{l'}}{M} }\|\frac{\Vec{l'}}{M}-\Vec{p}\|_1}.
    \eal
\end{defn}
If a probability distribution belongs to several blocks, we choose a single $l$ out of them.
Note that not all the blocks contain the valid type classes.
From an observation, the following lemma holds.
\begin{lem}
For any probability distribution $\Vec{p},\vec{q}\in\mbP_d$ and for any ${\Vec{l}}$, it holds that
    \bal
    \Vec{p},\Vec{q}\in\mB(\Vec{l})\Rightarrow \|\Vec{p}-\Vec{q}\|_1\leq\frac{D}{M}.
    \eal
\end{lem}

We define a projector onto the subspace corresponding to these boxes as follows.
\bal
P_{\mB(\Vec{l})}=\sum_{{\bold s}|{\rm Type} ({\bold s})\in\mB(\Vec{l})}\ketbra{\bold s}{\bold s}.
\eal

In the following discussion, we construct the universal work extraction protocol. 
We construct the Hamiltonian of the battery system as follows.
\bal
H_X^M=\sum_{\Vec{l}~|~\Tr\qty[P_{\mB(\Vec{l})}\tau^{\otimes n}  ]\neq 0} \frac{1}{\beta}\log\frac{1}{\Tr\qty[P_{\mB(\Vec{l})}\tau^{\otimes n}  ]} \ketbra{W_{\Vec{l}}}{W_{\Vec{l}}}_X.
\eal
Note that the number of the vectors $\Vec{l}$ which satisfy $\Tr\qty[P_{\mB(\Vec{l})}\tau^{\otimes n}  ]\neq 0$ is at most finite; it suffices to take a finite-dimensional battery system.
Our protocol is described as follows.
\bal
\mE_M(\rho_{n})=\sum_{\Vec{l}}\Tr[P_{\mB(\Vec{l})}\rho_{n}]\ketbra{W_{\Vec{l}}}{W_{\Vec{l}}}_X.
\eal
By design, $\mE_M$ is Gibbs-preserving. 
We apply this channel to the state $\rho^{\otimes n}$. Note that there exist infinitely many $M$ such that $\Vec{p}$ is an interior point of $\mB(\Vec{l})$.
For such $M$, there exist a positive number $r>0$ which satisfies $\mB_r(\Vec{p})\subset \mB(\Vec{l})$. Here, $\mB_r(\Vec{p})$ is the open ball with respect to the total variation distance. We can take sufficiently large $n$ such that 
\bal
\Tr[P_{\mB_\ve(\Vec{l})}\rho^{\otimes n}]> 1-\zeta
\eal
holds for any $\zeta>0$.
Note that if the state is exactly on the boundary of the blocks, the inequality above does not hold.

From Sanov's theorem~\cite{cover_1999_elements}, it holds that
\bal
\frac{1}{\beta}\lim_{n\to\infty}\frac{1}{n}\log\frac{1}{\Tr\qty[P_{\mB(\Vec{l})}\tau^{\otimes n}]}= \frac{1}{\beta}D(\Vec{p^*}\|\Vec{q}),\\
\vec{p^*}=\arg\min_{\Vec{p'}\in\mB(\Vec{l})} D(\Vec{p'}\|\Vec{q}).
\eal
Due to Lemma~\ref{Lem: continuity bound}, it holds that
\bal
\abs{D(\rho_1\|\tau)-D(\rho_2\|\tau )}
&\leq\qty(1 +\frac{\beta E_{\rm max}+\log Z}{\sqrt{2}}) \| \rho_1-\rho_2 \|_1
\eal
We can take sufficiently large $M$ independent of the initial state so that we have $\abs{D(\Vec{p^*}\|\Vec{q })-D(\Vec{p}\|\Vec{q})}\leq \delta$.
From the discussion above, we can conclude that if we take a sufficiently large number of samples and sufficiently large $M$, we can achieve the work extraction rate $D(\vec{p}\|\Vec{q})=D(\rho\|\tau)$ with any accuracy.

Combining these, together with the Schur pinching applied before the classical protocol described in this subsection, we have the following result. For a family of the protocols $\qty{\qty{\Lambda^{M}_n}_{n\in\mbN}}_{M\in\mbN}$ we construct, it holds that
\bal
\forall \ve, \delta>0, ~\exists~M\in\mbN~~&{\rm s.t.}~~\forall\rho\in\mD(\mH)~{\rm (a.e.)},~\exists N_\rho\in\mbN~~{\rm s.t.}~~n\geq N_\rho \\&\Rightarrow F(\Lambda^{M}_n(\rho^{\otimes n}),\ketbra{n(D(\rho\|\tau)-\delta)}{n(D(\rho\|\tau)-\delta)})\geq 1-\ve.
\eal
Here, (a.e.) implies that the inequality holds for almost all density matrices. The claim is weakened because the protocol does not work when the input state is precisely on the boundary of the block.
The state dependence is imposed on the sample complexity. Thus, the experimenter does not know how many copies one needs to use.

\subsubsection{Strategy based on state tomography}
 
In the protocol presented in the main text, we introduced the Schur pinching technique to overcome the difficulty of not knowing the eigenbasis that diagonalizes the state in the energy eigensubspaces.
Alternatively, one can employ a more ``brute-force'' approach, where we learn the description of the density matrices by state tomography. 
What is nontrivial here is that the thermal operations do not contain measurements.
We solve this by using the incoherently conditioned thermal operations, which can be implemented by thermal operations because of Lemma~\ref{app lem: TO=ICPTO}.
Namely, we measure a sublinear number of copies of the initial state to estimate the description of the initial state and apply the appropriate thermal operation to the remaining copies depending on the estimate.

One could then notice that this approach comes with another subtlety. 
To ensure this conditioned thermal operation to be implementable by a thermal operation, the measurement needs to be incoherent, while incoherent measurements are not informationally complete and are generally not sufficient to get the classical description of the state $\rho$. 
In the following, we see that it is not a problem for us, because what matters for the performance of our protocol is the description of the quasi-classical state that is block-diagonal in the energy eigenspaces.
The idea of the construction is described as follows.
\begin{enumerate}
    \item Given $\rho^{\otimes n}$, apply the pinching channel in Eq.~\eqref{app eq: pinching} to each $k$ copies.
    \item Perform the tomography in the energy subspace and estimate $\mP(\rho^{\otimes k})$. Let $\hat{\rho_k}$ be the estimated density matrix. This procedure can be implemented in the thermal state in the following step due to Lemma~\ref{app lem: TO=ICPTO}.
    \item Apply the complete decohering channel $\Delta_{\hat{\rho_k}}(\mP(\rho^{\otimes k}))$ with respect to the eigenbasis of the estimated state of $\hat{\rho_k}$. 
    \item Apply the work extraction protocol according to the estimated state $\hat{\rho_k}$.
\end{enumerate}

Let us discuss the relation between the sample cost and the error $\| \Delta_{\hat{\rho_k}}(\mP(\rho^{\otimes k}))-\hat{\rho_k} \|_1$. From the contractility and the relation $\Delta_{\hat{\rho_k}}(\hat{\rho_k})=\hat{\rho_k}$, we have
\bal
\| \Delta_{\hat{\rho_k}}(\mP(\rho^{\otimes k}))-\hat{\rho_k} \|_1\leq \| \mP(\rho^{\otimes k})-\hat{\rho_k} \|_1.
\eal
With this strategy, we can show that the completely dephasing channel tailored in this discussion does not lose too much coherence, which can be shown as follows. First, observe that
\bal
\|  \Delta_{\hat{\rho_k}}(\mP(\rho^{\otimes k}))- \mP(\rho^{\otimes k}) \|_1 &\leq \|\Delta_{\hat{\rho_k}}(\mP(\rho^{\otimes k}))-\hat{\rho_k}\|_1+\|\hat{\rho_k}-\mP(\rho^{\otimes k}) \|_1\\
&\leq \|\hat{\rho_k}-\mP(\rho^{\otimes k}) \|_1+\|\hat{\rho_k}-\mP(\rho^{\otimes k}) \|_1=2\|\hat{\rho_k}-\mP(\rho^{\otimes k}) \|_1,
\eal
and thus the following holds.
\begin{equation}
    \begin{aligned}
        &\abs{\frac{1}{k}D(\Delta_{\hat{\rho_k}}(\mP(\rho^{\otimes k}))\|\tau^{\otimes k})- D(\rho\|\tau) } \\&\leq \abs{\frac{1}{k}D(\Delta_{\hat{\rho_k}}(\mP(\rho^{\otimes k}))\|\tau^{\otimes k})-\frac{1}{k}D(\mP(\rho^{\otimes k})\|\tau^{\otimes k})}+\abs{\frac{1}{k}D(\mP(\rho^{\otimes k})\|\tau^{\otimes k})-D(\rho\|\tau)}\\
&\leq \frac{2}{k}\qty(1+k\frac{ \beta E_{\max} +\log Z }{\sqrt{2}})\|\hat{\rho_k}-\mP(\rho^{\otimes k}) \|_1+\frac{1}{k}\log (k+1)^{d-1}
    \end{aligned}
\end{equation}

Furthermore, since the estimated state $\hat{\rho_k}$ is incoherent state because the available measurement is restricted to the incoherent measurement, the completely dephasing channel $\Delta_{\hat{\rho_k}}$ is a thermal operation, which is a direct consequence of Lemma~\ref{lem: pinching is thermal}.

As long as $k$ goes to infinity and the error $\|\hat{\rho_k}-\mP(\rho^{\otimes k}) \|_1$ on the tomography goes to 0 in the limit $n\to\infty$, the loss of the relative entropy due to applying $\Delta_{\hat{\rho_k}}$ is negligible.
Thus, we can concentrate on the sample cost of tomography.
Here, we have to be careful that what we can perform is restricted to the incoherent measurement.
If we concentrate on the eigenspaces of $H^{\times k}$, incoherent measurement is informationally complete for each energy eigenspace.
Roughly, the sample cost of the tomography to estimate the density matrix of $\mP(\rho^{\otimes k})$ with the accuracy $\ve$ with probability $1-\delta$ is evaluated as
\bal\label{app eq: sample cost in tomographic approach}
\sum_{E_k\in\spec(H^{\otimes k})}\mO\qty(\frac{1}{p_{E_k}}  \frac{d_{E_k}}{\ve^2}\log \frac{1}{\delta})\approx \mO\qty(\frac{1}{2^{-kS(\rho)}}   \frac{2^{kS(\rho)} }{\ve^2}\log \frac{1}{\delta}  )\lesssim\mO\qty(\frac{d^{3k}}{\ve^2}\log\frac{1}{\delta})
\eal
Therefore, we can take the parameter $k_n$ properly so that we can construct the universal work extraction in a naive way.

Let us make a few comments about the tomography-based approach. 
We first stress that whether this ``learn-and-run'' strategy can achieve the same asymptotic work extraction rate is far from obvious a priori.   
This is because of the tradeoff between the copies of states for tomography and the accuracy of the estimation.
To achieve the optimal work extraction rate, we are allowed to consume only a sublinear amount of input states for the state tomography, which gives a limitation in the accuracy of the estimate.
On the other hand, the state-aware protocol subsequently applied after the learning might need a greater accuracy, i.e., better scaling with $n$, in which case the fidelity error can grow up to 1 in the limit $n\to\infty$.
Whether this problem arises or not depends on the specifics of the protocol and its performance analysis, and what our result shows precisely is that one can design a work extraction protocol that can avoid this problem.

Secondly, seeing this alternative protocol that does not use the Schur pinching, one might wonder the advantage of introducing the Schur pinching in the first place.
To see this, we recall that the strength of the Schur pinching lies in the observation that the permutation symmetry allows us to avoid the estimation of the eigenbasis. 
We find in Section~\ref{app sec: finite copy} that this operational advantage of considering the Schur pinching becomes most apparent by considering the finite-copy regime and when the experimenters are only informed of the Helmholtz free energy $D(\rho\|\tau)$ of the input state, while not knowing the full description of the input state itself.
This is one of the standard settings in the Shannon theory, in the context of the fixed-length universal coding~\cite{Josza1998universal, Hayashi_2009}.
Beyond the practical advantage we discuss in the next section, we foresee that the idea of Schur pinching will bring further useful information-theoretic insights into quantum thermodynamics and other related operational settings.

\subsection{Performance of the universal work extraction in the finite-copy regime}\label{app sec: finite copy}
In the discussion so far, we showed that the optimal work extraction rate converges to relative entropy without the knowledge of the input state in the limit $n \to\infty$, where $n$ represents the number of copies of input states.
Here, it is insightful to see how the work extraction rate behaves in the finite-copy regime.
As is already discussed in Section~\ref{App Sec: Brandao}, the work extraction rate from finite copies of input states in the state-aware scenario is evaluated as
\bal
\frac{1}{n}\beta W^\ve_{\rm aware}(\rho^{\otimes n})\sim D(\rho\|\tau)-\mO\qty(\frac{1}{\sqrt{n}}).
\label{eq:state aware convergence}
\eal

Let us now discuss the performance of the fully universal work extraction protocol in the finite-copy regime. Without any knowledge of the input state, the situation is a lot more complicated due to the necessity of the estimation procedure.
Let us first consider the universal work extraction protocol with Schur pinching.
Since we need to guarantee that the final state is $\ve$-close to the pure excited state of the work storage, due to Lemma~\ref{Lem: Hoeffding inequality}, $r_n$ satisfies 
\bal
r_n\geq \sqrt{ \frac{d^{k_n}\ln 2+\ln \frac{2}{\ve}}{2m_n}  }.
\eal
From Eq.~\eqref{app eq: guaranteed performence after pinch} and Eq.~\eqref{app Eq: bound for the schur pinching}, it holds that the work extraction rate
\bal
\frac{1}{n}\beta \sum_i h(i)E_i\leq & \frac{k_n(q_n-m_n)}{n}\left[D\qty(\rho\|\tau)-3\qty(\frac{1}{k_n}+\frac{\beta E_{\rm max}+\log Z}{\sqrt{2}})\sqrt{ \frac{d^{k_n}\ln 2+\ln \frac{2}{\ve}}{2m_n}  } \right. \\
&\left. -\frac{1}{k_n}\log {\rm poly}(k_n)-\mO\qty(\frac{1}{\sqrt{k_n(q_n-m_n)}})\right].
\eal
is possible. 
If we take the same parameters as Eq.~\eqref{app eq: choice of parameters in finite}, we have 
\bal\label{app eq: finite copy performance}
\frac{1}{n}\beta \sum_i h(i)E_i&\sim \qty(1-\frac{1}{n^{\alpha}})\qty[D(\rho\|\tau)-\mO\qty(\frac{\log\log n}{\log n})]\\
 \frac{1}{n}\beta W^\ve(\rho^{\otimes n}, \Lambda_n^{\rm univ})&\gtrapprox \qty(1-\frac{1}{n^{\alpha}})\qty[D(\rho\|\tau)-\mO\qty(\frac{\log\log n}{\log n})]
\eal
Here, $\alpha$ is a constant with $0<\alpha<\frac{1}{3}$.
We can see that the right-hand side can be arbitrarily close to the relative entropy $D(\rho\|\tau)$ in the asymptotic limit $n\to\infty$, which is consistent with Theorem~\ref{thm:universal}.
The coefficient $\qty(1-\frac{1}{n^\alpha})$ and the second term on the right-hand side do not appear in the evaluation of the work extraction protocol in the state-aware scenario.
Specifically, the second term is because of $\frac{1}{k_n}\log {\rm poly}(k_n)$, which comes from the Schur pinching.
Although the universal work extraction protocol introduced in this paper works as well as the state-aware protocol in the asymptotic limit, the performance for the finite copies of the input states is worse than that of the state-aware one. 
This signifies the disadvantage of the universal work extraction protocol compared to the state-aware one due to the inaccessibility of the information of the input state.

Now, let us consider the performance of the universal work extraction protocol with the tomographic approach. Because the scaling of the sample complexity in the tomographic approach is almost the same as that in the approach employing the Schur pinching up to some constant multiple, the behavior of the extractable work rate of the universal protocol with the tomographic strategy in the finite-copy regime is essentially the same as that of the strategy with Schur pinching. This implies that, in the fully state-agnostic scenario, Schur pinching does not give us an advantage compared to the tomographic approach, which is because the estimation procedure is the bottleneck of both strategies after all.

\subsubsection{Universal work extraction with partial information of free energy}

On the other hand, consider the scenario where we are given the value of relative entropy $D(\rho\|\tau)$ (while still not knowing the full description of $\rho$).
This assumption of knowing the information of the entropic quantity is often employed in the context of the universal quantum information processing protocols in quantum Shannon theory, such as the universal quantum source compression~\cite{Josza1998universal} and the universal classical-quantum channel coding~\cite{Hayashi_2009}.

In this case, we can construct the universal work extraction in a much simpler way if we use the Schur pinching, where we can entirely skip the tomography step, either for estimating the state description or relative entropy.
The construction goes similarly to the state-aware one exhibited in Section~\ref{app sec: general state state aware}. 
Suppose that we are given $n$ copies of the input state $\rho$, and we are informed of the value $D(\rho\|\tau)$.
First, we apply the Schur pinching $\tilde{\mP}$ to the whole system $\rho^{\otimes n}$. Now, we have the pinched state $\tilde{\mP}(\rho^{\otimes n})$ which is diagonalized by the energy eigenbasis which respects the structure of Schur-Weyl duality. 
All we have to do is to find the appropriate energy-conserving unitary to achieve the optimal work extraction. Recalling \eqref{app eq: existence of work extraction protocol for general states}, we can find the appropriate function $h(\cdot)$ with the knowledge of the relative entropy $D(\rho\|\tau)$. 
Specifically, the behavior of the extractable work rate in the finite-copy regime is written as
\begin{equation}
    \frac{1}{n}\beta W^\ve (\rho^{\otimes n}, \Lambda_n^{D(\rho\|\tau)})\sim D(\rho\|\tau)-\mO\qty(\frac{1}{\sqrt{n}}).
\end{equation}
Here, the superscript of $\Lambda_n^{D(\rho\|\tau)}$ implies that the protocol can depend on the relative entropy.
Comparing this to \eqref{eq:state aware convergence}, we can see that, under the knowledge of the relative entropy, Schur pinching enables us to achieve the  optimal work extraction not only in the sense of the leading order but also in the sense of the subleading order in the known state-aware protocol.

This strategy with the Schur pinching allows the extractable work per copy to converge as fast as the state-aware protocol as characterized in \eqref{eq:state aware convergence}, which is much faster than the one in \eqref{app eq: finite copy performance} realized by the tomographic approach without the Schur pinching. 
This comes from that the Schur pinching allows us to avoid the tomographic procedure, which is necessary in the tomographic approach due to the ignorance of the diagonal basis.

\subsection{On pseudo-nonequilibrium state} \label{app:pseudo}

Let us remark on the connection between the state-agnostic work extraction and the notion called the \emph{pseudo-resource}.
The pseudo-resource is a state ensemble consisting of states with smaller resource contents but cannot be distinguished from the one with a large amount of resources by any efficient protocol.
Previous works found that pseudo-entangled and pseudo-magic states can be constructed using the pseudo-random states~\cite{aaronson_2023_quantumpseudoentanglement,Gu_pseudomagic}, which cannot be distinguished efficiently from the Haar random states, which come with almost the maximum amount of resources with high probability.
The existence of an efficient universal resource distillation task is a crucial problem along this line because such a protocol can distinguish between the pseudo-entangled state and a more resourceful state by applying the universal resource distillation protocol and counting the distilled state. In fact, the existence of the pseudo-entangled and pseudo-magic states restricts the performance of the efficient resource distillation protocol in state-agnostic scenarios~\cite{Gu_pseudomagic,arnonfriedman2023computationalentanglementtheory,leone_entanglement_theory_limited_computational}.

In that sense, our work could serve as a possible direction to determine whether or not a ``pseudo-nonequilibrium states'' exist.
We first observe that the same construction based on the pseudo-random states does not work in the case of quantum thermodynamics because the Haar random states do not possess high average free energy. 
To see this, let us calculate the average of the Helmholtz free energy $F(\rho):=\Tr[\rho H]-TS(\rho)$ of $n$-qubit Haar random states, where $S(\rho)\coloneqq-\Tr\rho\log\rho$ is the von Neumann entropy.
Since any Haar random state $\psi$ is pure, $S(\psi)=0$ holds. Thus, it suffices to focus on the expectation value of the energy. Taking the average, we have
\bal
\int \dd \psi \Tr[H^{\times n}\psi]=\Tr[H^{\times n}\int \dd \psi~\psi]=\Tr[H^{\times n}\frac{I}{2^n}]=\frac{\Tr[H^{\times n}]}{2^n}.
\eal
Since $\Tr[H^{\times n}]=\mO(n)$, the average of the free energy is $\mO\qty(\frac{n}{2^n})$, which can be arbitrarily small.
 This observation suggests that it is reasonable to conjecture that pseudo-nonequilibrium states do not exist.
Although our protocol is not shown to be runtime efficient, or our finite-copy bound \eqref{app eq: finite copy performance} does not guarantee the same work extraction rate with only polynomially many input samples, i.e., $n=\mO(\log \dim\mH)$, improving our protocol to an efficient one would disprove the existence of pseudo-nonequilibrium states.

\section{Comparison with the universal entanglement distillation protocol}\label{App: matsumoto hayashi}
The work extraction task can be viewed as a special case of a central task in quantum resource theories, known as the resource distillation task, which aims to extract a resourceful state from a noisy resource state.
Here, the natural question to ask is whether other resource distillation, such as entanglement distillation and magic state distillation, can be accomplished universally.
Indeed, in Ref.~\cite{matsumoto_universal_distortion}, they constructed a universal entanglement distillation protocol that can achieve the optimal distillation rate for any pure bipartite entangled state by applying the local operations to each share.
Here, let us briefly see how this universal entanglement distillation protocol is constructed and clarify how it differs from our universal work extraction protocol. 

Suppose that one is given $n$ copies of a bipartite pure state $\ket{\psi}_{AB}\in \mH_A\otimes \mH_B$, where we can assume $\mH_A\cong\mH_B$ without loss of generality.
Here, if we take the orthogonal basis of each space as $\qty{\ket{i}_A}$ and $\qty{\ket{j}_B}$, we can always obtain the decomposition $\ket{\psi}_{AB}=\sum_{ij}c_{ij}\ket{i}_A\otimes \ket{j}_B$.
Since taking partial transposition on Bob's share is isomorphic, we consider the matrix $\tilde{\psi}=\sum_{ij}c_{ij}\ketbra{i}{j}$ for a while.
Because of Schur-Weyl duality $\mH^{\otimes n}=\bigoplus_{\lambda\in Y^n_d}\mW_\lambda\otimes \mU_\lambda$, $\tilde{\psi}^{\otimes n}$ can be decomposed as 
\bal
\tilde{\psi}^{\otimes n}=\bigoplus_{\lambda\in Y^n_d}\tilde{\psi}_\lambda\otimes I_\lambda.
\eal
Taking the transposition again, we obtain the following decomposition.
\bal
\ket{\psi}_{AB}^{\otimes n}=\sum_{\lambda\in Y^n_d}\ket{\tilde{\psi}_\lambda}_{AB}\otimes \ket{I_\lambda}_{AB}
\eal
Here, $\ket{\tilde{\psi}_\lambda}$ is an unnormalized vector in $\mW_{A, \lambda}\otimes \mW_{B, \lambda}$ and $\ket{I_\lambda}\in\mU_{A,\lambda}\otimes\mU_{B,\lambda}$ is the vector proportional to the maximally entangled state.
All one has to do is to perform the projective measurement $\qty{\Pi_{\mW_{i,\lambda}\otimes \mU_{i,\lambda}}  }_{\lambda\in Y^n_d}$ on each share $i=A, B$ and discard the subsystem which corresponds to the $\mW_\lambda$.
Also, by employing Weyl's dimension formula, one can see that the entanglement distillation rate in this strategy converges to the entanglement entropy $E(\psi_{AB}):=-\Tr[\psi_A\log \psi_A]$ of the initial state $\ket{\psi}_{AB}$. 
The key idea behind this protocol is that the matrix $\tilde{\psi}^{\otimes n}$ acts trivially on the irreps $\mU_\lambda$ of the symmetric group, and the identity matrices on $\mU_\lambda$ turn into the maximally entangled state due to the transposition.

One might think that this approach can be directly applied to the work extraction task, partially because both the necessary and sufficient conditions for the conversion of quantum states in the thermodynamic scenario and the bipartite pure states with local operations and classical communications are given by the mathematics of majorization.
One can easily see that if we have $n$ copies of $\rho$, this state admits the following decomposition.
\bal
\rho^{\otimes n}=\bigoplus_{\lambda\in Y^n_d}\tilde{\rho}_\lambda\otimes \frac{I_\lambda}{d_\lambda}.
\eal
Even though $\rho^{\otimes n}$ admits a similar decomposition, the identity in this decomposition does not contribute much to the work extraction because $I_\lambda/d_{\lambda}$ is close to the thermal state.

This observation indicates that the strategy to construct the universal entanglement distillation protocol does not help directly to tailor the universal work extraction protocol.
This is partially because the mazorization relation which governs the conversion of the pure bipartite state and the state in the thermodynamic system is opposite.
The universal work extraction protocol exhibited in this paper is thus constructed by carefully focusing on the specific properties of the resource theory of thermodynamics.

\section{Infinite-dimensional systems (Proof of Theorem~\ref{thm: inf universal})}\label{App sec: Infinite}
We consider the infinite-dimensional Hilbert space $\mH$, whose basis can be taken as the energy eigenbasis $\qty{\ket{i}}_{i\in\mbN}$ of the Hamiltonian $H=\sum_{i\in\mbN}E_i\ketbra{i}{i}$. 
We assume that the energy eigenvalues satisfy $E_i\leq E_{i+1}$ for any $i\in\mbN$ and that the Hamiltonian satisfy the Gibbs hypothesis 
\bal
\Tr[e^{-\beta H}]<\infty
\eal
so that the thermal state $\tau = e^{-\beta H}/\Tr[e^{-\beta H}]$ is well defined.

\subsection{State-aware work extraction protocol for infinite-dimensional system}
Let us first consider the situation where we are given the full description of the given state $\rho^{\otimes n}$ in the infinite-dimensional Hilbert space.
Specifically, we focus on the state whose diagonal elements decay as follows: there exists a positive number $\ve>0$ which satisfies $\bra{i}\rho\ket{i}=\mO(i^{-(2+\ve)})$. This assumption is almost equivalent to the condition $\Tr[\rho H]<\infty$ when the system under consideration is the harmonic oscillator. When the Hamiltonian grows superlinearly, i.e., $E_i=\Omega(i)$, the finite-energy condition is sufficient to satisfy this assumption. 

One might think that we could apply the same work extraction protocol for the finite-dimensional Hilbert space.
However, the analysis cannot directly be applied in this case because the cardinality of the set defined in Eq.~\eqref{eq: def of freq} is not finite, and the analysis in Eq.~\eqref{app eq: condition for injection} is not valid.
To circumvent this problem, we apply the work extraction protocol in Ref.~\cite{Brandao_resource_theory_of} to the finite-dimensional subspace.
To determine which subspace to apply the protocol, we set the cutoff dimension $d_n$, which can depend on the number $n$ of copies.
We define $\Pi_{d_n}:=\sum_{i=1}^{d_n}\ketbra{i}{i}$ and $\Pi_\perp:=I-\Pi_{d_n}$. 
Now, the coarse-grained pinching channel 
\bal
\mP'_{d_n}(\rho):=\Pi_{d_n}\rho\Pi_{d_n}+\Pi_\perp\rho\Pi_\perp
\eal
is a thermal operation, because the convexity of the thermal operation can be easily extended to the infinite-dimensional system. When this channel is applied, this state is block-diagonalized as 
\bal
\mP'_{d_n}(\rho)=\rho_{d_n}\oplus \rho_\perp.
\eal
$\qty(\mP'_{d}(\rho))^{\otimes n}$ is also block-diagonalized as
\bal
\qty(\mP'_{d_n}(\rho))^{\otimes n}=(\rho_{d_n})^{\otimes n}\oplus \rho^n_\perp, ~~\rho^n_\perp\in\mD((\mH_{d_n}^{\otimes n})^\perp).
\eal
Here, $\mH_d$ is defined as $\mH_d=\Span\qty{\ket{1},\ldots,\ket{d}}$.
Our strategy is to apply the work extraction protocol to the subnormalized state $(\rho_{d_n})^{\otimes n}$, which can be interpreted as the protocol succeeds probabilistically. 
The probability of failure corresponds to the fidelity error of the excited state of the work storage.
The success probability is represented as $\Tr\qty[(\rho_{d_n})^{\otimes n}]=\Tr[\rho_{d_n}]^n$.
Since we assume that the diagonal elements of $\rho$ scales as $\bra{i}\rho\ket{i}=\mO(i^{-(2+\ve)})$, we can take $d_n:=n^{1/(1+\ve/2)}$. It can be checked easily that $\Tr[\rho_{d_n}]^n\to 1$ in the limit $n\to\infty$.

Now, we consider the optimal performance of the work extraction protocol explained in Appendix~\ref{App Sec: Brandao}.
We apply the pinching channel to $k_n=n^{1/(1+\ve/4)}$ subsystems. After this, we obtain $n/k_n$ copies of $\mP(\rho_{d_n}^{\otimes k})\oplus \rho_{n,\perp}$.
Due to the construction, the optimal extractable work is characterized by the relative entropy of the renormalized input state with respect to the renormalized thermal state
\bal
\frac{1}{n}\beta W^\ve_{\rm aware}(\rho^{\otimes n})\geq D\qty(\mP\qty(\qty(\frac{\rho_{d_n}}{\Tr[\rho_{d_n}]})^{\otimes k_n})\Big\|\qty(\frac{\tau_{d_n}}{\Tr[\tau_{d_n}]})^{\otimes k_n})-\mO\qty(\frac{1}{\sqrt{n/k_n}}).
\eal
This quantity can be arbitrarily close to the relative entropy of the renormalized input state with respect to the renormalized thermal state, which can be checked as follows.
\begin{equation}
    \begin{aligned}
        &\abs{D\qty(\frac{\rho_{d_n}}{\Tr[\rho_{d_n}]}\Big\|\frac{\tau_{d_n}}{\Tr[\tau_{d_n}]})-\frac{1}{k_n}D\qty(\mP\qty(\qty(\frac{\rho_{d_n}}{\Tr[\rho_{d_n}]})^{\otimes k_n})\Big\|\qty(\frac{\tau_{d_n}}{\Tr[\tau_{d_n}]})^{\otimes k_n})}=\frac{1}{k_n}D\qty(\qty(\frac{\rho_{d_n}}{\Tr[\rho_{d_n}]})^{\otimes k_n}\|\mP\qty(\qty(\frac{\rho_{d_n}}{\Tr[\rho_{d_n}]})^{\otimes k_n}))\\
&\leq\frac{1}{k_n}\log (k_n+1)^{d_n}=\frac{1}{n^{2\ve/(2\ve+1)(4\ve+1)}}\log(n^{\frac{1}{1+\ve/4}}+1)\to 0
    \end{aligned}
\end{equation}
Here, we employed Hayashi's pinching inequality~\cite{hayashi_optimal_2002}.

Moreover, note that we have
\bal
D\qty(\frac{\rho_d}{\Tr[\rho_d]}\Big\|\frac{\tau_d}{\Tr[\tau_d]})&=\Tr\qty[\frac{\rho_d}{\Tr[\rho_d]}\log\frac{\rho_d}{\Tr[\rho_d]}-\frac{\rho_d}{\Tr[\rho_d]}\log\frac{\tau_d}{\Tr[\tau_d]}]\\
&=\frac{1}{\Tr[\rho_d]}D_L(\rho_d\|\tau_d)+\frac{\Tr[\rho_d]-\Tr[\tau_d]}{\Tr[\rho_d]}+\log\frac{\Tr[\tau_d]}{\Tr[\rho_d]}\\
&\xrightarrow{d\to\infty}D(\rho\|\tau), 
\eal
where $D_L(\rho\|\sigma)=\Tr[\rho\log\rho-\rho\log\sigma]+\Tr[\sigma]-\Tr[\rho]$ is Lindblad extension of the relative entropy~\cite{Lindblad1974expectations}, and the last line holds because
\bal
\lim_{d\to\infty}\Tr[\rho_d]=\lim_{d\to\infty}\Tr[\tau_d]=1,~~D_L(\rho_d\|\tau_d)=D(\rho\|\tau).
\eal
The last equality is the result from Ref.~\cite[Lemma 4]{Lindblad1974expectations}.

In conclusion, the extractable work rate is characterized by the relative entropy, which gives us the state-aware work extraction protocol for states in an infinite-dimensional system.

Conversely, this extractable work rate is optimal.
In Ref.~\cite{ferrari_asymptotic_2023}, they gave the upper bound of the conversion rate between two states $\rho,\sigma$ utilizing a resource measure with some specific properties in the general resource theory, which allows us to characterize the fundamental limitation of the extractable work from a given state on an infinite-dimensional Hilbert space.

For the reader's convenience, we summarize the converse bound obtained in Ref.~\cite{ferrari_asymptotic_2023}.
As mentioned in the main text, quantum resource theory is characterized by the set of free states $\mbF(\mH_A)\subset \mD(\mH_A)$ defined for arbitrary Hilbert spaces, which includes the states prepared without any cost, and the class of free operations $\mbO(\mH_A\to\mH_B)\subset {\rm CPTP}(\mH_A\to\mH_B)$, the class of CPTP maps, which is supposed to be applied freely in the given situation. The constraint for these two is that any free operations should map the free states to the free states, namely $\Lambda(\mbF(\mH_A))\subset \mbF(\mH_B), \forall \Lambda\in\mbO(\mH_A\to\mH_B)$, which implies that the operation applied without any cost in the given scenario should not create a resourceful state.
In the resource theory of thermodynamics, the thermal state is the only free state. The largest class of free operations is called Gibbs-preserving operations, which contains all the CPTP maps that map the thermal state of the input system to the thermal state of the output system.

How much a given state possesses the resource we are interested in is measured via functions with several properties called \emph{resource monotone}.
\begin{defn}
    Let $\mbF$ be the set of free states and $\mbO$ be the class of the free operations of a quantum resource theory.
    A function $G:\mD(\mH)\to\mbR_{\geq 0}$ defined on the set of density matrices of any Hilbert space $\mH$ is called a resource monotone if and only if $G$ satisfies the following conditions.
    \begin{enumerate}
        \item For any pair of Hilbert spaces $\mH_A, \mH_B$, $\Lambda\in\mbO(\mH_A\to\mH_B)$, and $\rho_A\in\mD(\mH_A)$, $G(\rho_A)\geq G(\Lambda(\rho_A))$ holds. (\emph{monotonicity})
        \item For any $\rho\in\mbF(\mH)$, it holds that $G(\rho)=0$.
    \end{enumerate}
    Moreover, A resource monotone $G$ is said to be
    \begin{enumerate}
        \item \emph{strongly superadditive} if $G$ satisfies
        \bal
        G(\rho_{AB})\geq G(\rho_A)+G(\rho_B)
        \eal
        for any $\rho_{AB}\in\mD(\mH_A\otimes \mH_B)$. Here, $\rho_A:=\Tr_B\rho_{AB}$ and $\rho_B:=\Tr_A\rho_{AB}$.
        \item \emph{weakly additive} if $G$ satisfies
        \bal
        G(\rho^{\otimes n})=nG(\rho)
        \eal
        for any $\rho\in\mD(\mH)$.
        \item \emph{lower semi-continuous} if for any $\rho\in\mD(\mH)$ and any sequence $\qty{\rho_n}\subset\mD(\mH)$ with $\lim_{n\to\infty}\|\rho-\rho_n\|_1=0$, $G$ satisfies
        \bal
        \liminf_{n\to\infty}G(\rho_n)\geq G(\rho).
        \eal
    \end{enumerate}
\end{defn}

In the resource theory of thermodynamics, the Helmholtz free energy $F(\rho)=\frac{1}{\beta}D(\rho\|\tau)$ is a resource monotone, which can be checked easily by $D(\tau\|\tau)=0$ and the data-processing inequality.

Now, we are interested in the optimal rate of the state transformation between two resource states with free operations in the asymptotic limit.
\begin{defn}
    Let $\mbF$ and $\mbO$ be the set of free states and the class of free operations of the quantum resource theory under consideration, and let us take two states $\rho_A\in\mD(\mH_A)$ and $\sigma_B\in\mD(\mH_B)$.
    The (standard) asymptotic transformation rate of the conversion from $\rho_A$ to $\sigma_B$ is defined as
    \bal
R(\rho_A\to\sigma_B)=\sup\qty{r~|~\limsup_{n\to\infty}\inf_{\Lambda\in\mbO(\mH_A^{\otimes n}\to\mH_B^{\otimes \lfloor rn\rfloor})} \| \Lambda(\rho_A^{\otimes n})-\sigma_B^{\otimes \lfloor rn\rfloor} \|_1=0 }.
    \eal
\end{defn}
Ref.~\cite{ferrari_asymptotic_2023} provided an upper bound of the transformation rate as follows.

\begin{pro}[{\cite{ferrari_asymptotic_2023}}]
    Suppose that the resource monotone $G(\cdot)$ of a quantum resource theory satisfies strong superadditivity, weak additivity, and lower-semicontinuity. The transformation rate between two states on the Hilbert spaces, which are not necessarily finite-dimensional, is bounded as
    \bal\label{app eq: ferrari lami}
    R(\rho_A\to\sigma_B)\leq \frac{G(\rho_A)}{G(\sigma_B)}.
    \eal
\end{pro}
Since Umegaki relative entropy satisfies these three conditions above, this result can also be applied to quantum thermodynamics by replacing $G(\cdot)$ with $D(\cdot\|\tau)$.

We remark on the connection between the optimal transformation rate and the optimal performance of the work extraction task from the known input state.
In fact, this connection is not clear via the formalization of the extractable work with the work storage mentioned in the main text.
Again, we utilize the notion of the work battery explained in Section~\ref{app sec: universal work extraction}.
We take the output system as the qubit system $\mH_B=\Span\qty{\ket{0},\ket{1}}$ with the trivial Hamiltonian $H=0$. 
Observe that
\bal
D_{\max}(\ketbra{1}{1}\|\tau_B)=D(\ketbra{1}{1}\|\tau_B)=D_{\min}(\ketbra{1}{1}\|\tau_B)=1.
\eal
Repeating the same discussion as the proof of Lemma~\ref{lem: work storage=work battery}, we can see that the extractable work from the initial state $\rho_A$ is proportional to the number of the excited state $\ketbra{1}{1}$ of the qubit system obtained by the free operations.
Due to this observation, it holds that
\bal
\beta W^{\infty}_{\rm aware}(\rho)=R(\rho_A\to\ketbra{1}{1}_B)D(\ketbra{1}{1}\|\tau_B)\leq D(\rho_A\|\tau_A),
\eal
and we obtain the upper bound of the extractable work rate.

This shows the converse bound $\beta W^{\infty}_{\rm aware}(\rho)\leq D(\rho\|\tau)$. 
Combining this with the direct part $\beta W^{\infty}_{\rm aware}(\rho)\geq D(\rho\|\tau)$ achieved by our protocol described above, we establish the following optimal work extraction rate.
\begin{thm}
    Let $\mH$ be an infinite-dimensional Hilbert space with Hamiltonian $H$ satisfying $\Tr(e^{-\beta H})<\infty$. 
    Suppose also that, for the given state $\rho$, there exists a positive number $\ve>0$ such that the diagonal elements of $\rho$ satisfy $\rho_{ii}=\mO(i^{-(2+\ve)})$. Then, it holds that
    \bal
    \beta W^{\infty}_{\rm aware}(\rho)=D(\rho\|\tau).
    \eal
\end{thm}

\subsection{Semiuniversal work extraction for infinite dimensions}

Let us now extend the previous state-aware protocol to a semiuniversal one, where an unknown input state is given from a set of i.i.d. states taken from a finite number of candidates.
Fix a subset $S\subset\mD(\mH)$ of states containing finite elements. This set represents the possible states to be given.
For any positive integer $n\in\mbN$, we define the set $S_n\in\mD(\mH^{\otimes n})$ of states as
\bal
S_n=\qty{\rho^{\otimes n}~|~\rho\in S\subset\mD(\mH),~\abs{S}<\infty}.
\eal
We assume that we are given the complete knowledge of $S$, and all the states in the set follow the scaling $\bra{i}\rho\ket{i}=\mO(i^{-(2+\ve)})$ for some $\ve$, which can depend on the state. We construct the work extraction protocol to achieve the optimal work extraction rate for each state in $S$.
In this sense, we call this protocol ‘semiuniversal.’

If there exists a cutoff dimension $d$ which satisfies $\Pi_d\rho\Pi_d=\rho$ for any states in $S$, that is, all the states in the set $S$ have finite-dimensional support, it suffices to apply the universal work extraction protocol for finite-dimensional systems to the subspace $\mH_d$. 
Thus, we can concentrate on situations where this does not hold.

First, we consider a pinching channel $\mP_{\tilde{H}}$ which corresponds to a virtual Hamiltonian $\tilde{H}^{\times n}$, where $\tilde{H}$ is  $\tilde{H}=\sum_{i=1}^\infty\tilde{E}_i\ketbra{i}{i}$ whose eigenvalues satisfy $\tilde{E}_i<\tilde{E}_{i+1}$. This pinching channel vanishes all the nontrivial degenerate subspaces in the original Hamiltonian.
Before starting the protocol, we specify the set
\bal
\mP_{\tilde{H}}(S_n):=\qty{\mP_{\tilde{H}}(\rho^{\otimes n})~|~\rho^{\otimes n}\in S_n}.
\eal
Here, there might exist two states, such as $\left[\frac{1}{\sqrt{2}}(\ket{1}\pm\ket{2})\right]^{\otimes n}$, whose pinched state coincides, and the number $\abs{\mP_{\tilde{H}}(S_n)}$ of elements in the pinched set might be less than that of the original set $\abs{S_n}$.
However, since the work extraction protocol discussed in the subsequent discussion only depends on the information about the energy subspace that survives after the pinching channel, it suffices to concentrate on the set $\mP_{\tilde{H}}(S_n)$.

If we can distinguish the given states from other states by the incoherent measurement using a sublinear amount of samples, we can tailor the work extraction protocol for this set of states.

Here, we prove the following lemma.
\begin{lem}
    Let $\rho,\sigma\in\mD(\mH)$ be two density matrices on the infinite-dimensional Hilbert space. Moreover, let $\Pi_{d,n} $ be the projector onto the subspace $(\mH_d)^{\otimes n}$. The following two conditions are equivalent.
    \begin{itemize}
        \item $\forall d\in\mbN,~~ \Pi_{d,d}\mP_{\tilde{H}}(\rho^{\otimes d})\Pi_{d,d}=\Pi_{d,d}\mP_{\tilde{H}}(\sigma^{\otimes d})\Pi_{d,d}$\\
        \item $\forall n\in\mbN, ~~\mP_{\tilde{H}}(\rho^{\otimes n})=\mP_{\tilde{H}}(\sigma^{\otimes n})$
    \end{itemize}
\end{lem}
\begin{proof}
    Key result to show this lemma is the following relation~\cite[Lemma S.22]{Watanabe_Black_box}: for any $\rho,\sigma\in\mD(\mH)$,
    \bal
    \Pi_{d,d}\mP_{\tilde{H}}(\rho^{\otimes d})\Pi_{d,d}=\Pi_{d,d}\mP_{\tilde{H}}(\sigma^{\otimes d})\Pi_{d,d}\Leftrightarrow \Pi_{d,n}\mP_{\tilde{H}}(\rho^{\otimes n})\Pi_{d,n}=\Pi_{d,n}\mP_{\tilde{H}}(\sigma^{\otimes n})\Pi_{d,n},~~\forall n\in\mbN.
    \eal
    From this, we have
    \bal
\forall d\in\mbN,~~ \Pi_{d,d}\mP_{\tilde{H}}(\rho^{\otimes d})\Pi_{d,d}&=\Pi_{d,d}\mP_{\tilde{H}}(\sigma^{\otimes d})\Pi_{d,d}  \\
\Leftrightarrow  \forall d,n\in\mbN,~~\Pi_{d,n}\mP_{\tilde{H}}(\rho^{\otimes n})\Pi_{d,n}&=\Pi_{d,n}\mP_{\tilde{H}}(\sigma^{\otimes n})\Pi_{d,n}\\
\Leftrightarrow~~\forall n\in\mbN~~\mP_{\tilde{H}}(\rho^{\otimes n})&=\mP_{\tilde{H}}(\sigma^{\otimes n}),
    \eal
    which concludes the proof.
\end{proof}
Therefore, if we try to distinguish between $\rho$ and $\sigma$ whose pinched states are different, we can concentrate on the subspace $(\mH_d)^{\otimes d}.$ We denote this specific dimension as $d_{\rho,\sigma}$.
Since $\abs{S}$ is at most finite, 
\bal
\tilde{d}=\max_{\substack{ \rho,\sigma\in S\\ \exists n\in\mbN,  \mP_{\tilde{H}}(\rho^{\otimes n})\neq \mP_{\tilde{H}}(\sigma^{\otimes n})   }}d_{\rho,\sigma}
\eal
is also finite. 

To tell apart the states in the set $S$ by the incoherent measurement, it suffices to perform the state tomography on $\Pi_{\tilde{d},\tilde{d}}\mP_{\tilde{H}}(\rho^{\otimes \tilde{d}})\Pi_{\tilde{d},\tilde{d}}$. 
Let us evaluate the sample complexity of this tomography.
Note that performing the tomography on the subnormalized state in the subspace $(\mH_{\tilde{d}})^{\otimes\tilde{d}}$ is equivalent to performing the tomography on the normalized state 
\bal
\tilde{\rho}_{\tilde{d}}:=\Pi_{\tilde{d},\tilde{d}}\mP_{\tilde{H}}(\rho^{\otimes \tilde{d}})\Pi_{\tilde{d},\tilde{d}}\oplus \qty(1-\Tr[\Pi_{\tilde{d},\tilde{d}}\mP_{\tilde{H}}(\rho^{\otimes \tilde{d}})\Pi_{\tilde{d},\tilde{d}}])\in\mD\qty((\mH_{\tilde{d}})^{\otimes \tilde{d}}\oplus \mbC).
\eal
Since $\abs{S}$ is finite, 
\bal
\tilde{\xi}:=\min_{\substack{ \rho,\sigma\in S\\ \exists n\in\mbN,  \mP_{\tilde{H}}(\rho^{\otimes n})\neq \mP_{\tilde{H}}(\sigma^{\otimes n})   }}\|\tilde{\rho}_{\tilde{d}}-\tilde{\sigma}_{\tilde{d}}\|_{1}
\eal
is positive and does not depend on $n$.
It suffices to evaluate the state with accuracy $\frac{\xi}{2}$.
In \cite{O'Donnell2016efficient}, it is shown that the sample complexity to estimate the detail of the density matrix allowing the error $\xi/2$ with respect to the trace distance with success probability $1-p_e$ is given by
\bal
\mO\qty(\frac{\qty(\tilde{d}^{\tilde{d}}+1)^2}{\xi^2/4}\log\qty(\frac{1}{p_e})).
\eal
Therefore, given the target error on fidelity, we need to guess the input system using at most a constant number of subsystems.
After the estimation, we can apply the work extraction protocol in the state-aware scenario. Since the number of subsystems consumed for the tomography is at most constant, this does not affect the work extraction rate.

Due to Lemma~\ref{app lem: TO=ICPTO}, the whole protocol can be implemented by a thermal operation. Thus, the protocol in the discussion above gives us a semiuniversal work extraction protocol for an infinite-dimensional system.

%TC:endignore
\end{document}